\documentclass[11pt]{article}

\usepackage[T1]{fontenc}
\usepackage{lmodern}
\usepackage[a4paper,margin=2.5cm]{geometry}

\usepackage{amsmath,amssymb,amsthm}
\usepackage{mathtools}

\usepackage{braket}
\usepackage{qcircuit}

\usepackage{graphicx}
\usepackage{placeins}
\usepackage{booktabs}

\usepackage{microtype}
\usepackage[hidelinks]{hyperref}

\newtheorem{theorem}{Theorem}[section]
\newtheorem{lemma}[theorem]{Lemma}

\newtheorem{corollary}[theorem]{Corollary}

\theoremstyle{definition}
\newtheorem{definition}[theorem]{Definition}

\theoremstyle{remark}
\newtheorem{remark}[theorem]{Remark}

\title{A Quantum Collocation Approach to One-Dimensional Boundary Value Problems with Coherent Amplitude Amplification}

\author{
Daniel Jaroszewski\textsuperscript{1,2}
\thanks{\texttt{daniel.jaroszewski@frankfurtconsultingengineers.de}}
\and
Bastian Harrach\textsuperscript{1}
\thanks{\texttt{harrach@math.uni-frankfurt.de}}
}

\date{}

\begin{document} \maketitle
    \begin{center} \small $^{1}$ Institute of Mathematics, Goethe University Frankfurt, Germany $^{2}$
        FCE Frankfurt Consulting Engineers GmbH, Frankfurt am Main, Germany \end{center}
        \begin{abstract}
            We propose a quantum collocation framework for approximating solutions of one-dimensional linear and
            nonlinear boundary value problems. The method formulates the search for admissible solutions as a
            residual-based quantum search over a discretized ansatz space, where candidate solutions are
            evaluated through residual conditions imposed at collocation points.

            A residual-threshold oracle is constructed that acts jointly on spatial and parameter registers.
            This joint oracle structure leads to amplification dynamics that decompose into a coherent
            superposition of spatially conditioned amplitude-amplification processes rather than a single global
            amplification mechanism.

            We derive the corresponding amplification geometry and show that the success probability is governed
            by a weighted combination of spatially dependent amplification angles. Furthermore, we prove that
            the reversible residual oracle can be implemented with gate complexity polynomial in the logarithm
            of the number of collocation points, while retaining the quadratic search acceleration associated
            with amplitude amplification in the parameter space.

            We analyze how the spatially dependent oracle structure influences the amplification dynamics and
            corresponding success probabilities. Furthermore, we investigate how discretization, ansatz
            expressivity, oracle tolerance, and finite-precision effects influence both approximation quality
            and amplification behavior. Numerical experiments validate the theoretical predictions and
            illustrate the resulting search dynamics across different discretization and precision regimes.
        \end{abstract}
        \section{Introduction}
        Boundary value problems (BVPs) for ordinary differential equations (ODEs) arise throughout applied
        mathematics, physics, and engineering. Classical numerical approaches such as shooting methods,
        finite-difference schemes, finite-element methods, and spectral discretizations provide
        well-established tools for solving such problems~\cite{AscherPetzold:1998,
        Leveque:2007,StoerBulirsch:2002, BrennerScott:2008, Hesthaven:2007}. In shooting methods, each
        residual evaluation requires solving an initial-value problem, with cost typically growing linearly
        with the number of integration steps. Finite-difference, finite-element, and spectral methods
        instead lead to algebraic systems whose cost depends on the number of degrees of freedom and the
        solver, ranging from linear to superlinear.

        Quantum algorithms for differential equations have often been developed around quantum linear-system
        algorithms (QLSAs), in which a discretized differential equation is reduced to a linearsystem and
        the solution is encoded in a quantum state. Foundational examples include the HHL
        algorithm~\cite{HHL2009} and later improvements by Childs, Kothari, and
        Somma~\cite{ChildsKothariSomma2017}. QLSA-based approaches have also been developed for
        finite-element discretizations and elliptic partial differential
        equations~\cite{MontanaroPallister2016, ChildsLiuOstrander2021}.

        Another direction uses Hamiltonian simulation or dynamical embeddings to represent differential
        equations as quantum dynamics. This includes quantum spectral methods~\cite{ChildsLiu2020},
        Schr\"odingerisation techniques~\cite{JinLiuYu2022, JinLiuYu2023}, and Hamiltonian-simulation-based
        solvers for partial differential equations~\cite{SatoKadowaki2024}. For nonlinear differential
        equations, quantum algorithms based on Carleman linearization and Koopman-type embeddings have also
        been developed~\cite{LiuEtAl2021, Krovi2023, JinLiuYu2023b}.

        More recently, physics-informed neural networks (PINNs) have provided a residual-based framework for
        approximating solutions of differential equations~\cite{Raissi:2019, DeRyck:2022,
        DeRyckMishra2024Survey}. Related quantum and hybrid quantum--classical approaches replace or augment
        the classical neural network by a parameterized quantum circuit~\cite{Berger2025, Panichi2026}.
        These methods retain the residual-minimization principle of PINNs and typically train the circuit
        parameters using a classical optimization procedure.

        In this work, we pursue a different approach based on collocation and quantum amplitude
        amplification. Instead of solving a discretized linear system or simulating a dynamical evolution
        problem, we formulatethe solution process as a residual-driven quantum search problem over a
        discretized ansatz space. Candidate solutions are evaluated through residual conditions imposed at
        collocation points, and amplitudeamplification is used to enhance parameter configurations
        satisfying a prescribed residual tolerance.

        A central feature of the proposed framework is that residual evaluation is performed coherently
        across both the collocation grid and the parameter space. This distinguishes the construction from a
        conventional Grover search in which each parameter configuration receives a single binary label.
        Instead, the oracle response reflects how the residual criterion is distributed across the spatial
        domain. Consequently, the amplification process retains information about where a candidate
        approximation satisfies or violates the differential equation.

        The numerical behavior of the method is governed by the collocation grid, ansatz expressivity,
        residual tolerance, and arithmetic precision. These factors determine both the attainable
        approximation accuracy andthe oracle structure seen by amplitude amplification. We analyze their
        influence theoretically and through numerical experiments.

        The remainder of the paper is organized as follows. Section~\ref{sec:problem_formulation} introduces
        the mathematical formulation of this special boundary value problem. Section~\ref{sec:overview}
        summarizes the main contributions of this work. Section~\ref{sec:quantum_collocation_method}
        develops the quantum collocation method, including the reversible residual oracle andthe associated
        amplification dynamics. Section~\ref{sec:oracle_cost_runtime} analyzes the oracle complexity and
        compares the resulting runtime behavior with standard Grover amplification.
        Section~\ref{sec:experimental_setup} and \ref{sec:experimental_analysis} present numerical
        experiments and implementation studies illustrating the influence of discretization, oracle
        tolerance, finite precision, and ansatz expressivity on the search dynamics. Finally, we conclude
        the paper in Section~\ref{sec:conclusion}.
        \section{Problem Formulation}
        \label{sec:problem_formulation} We consider boundary value problems
        (BVPs) of the form
        \begin{equation}
            u''(x) + f(x,u(x)) = 0, \quad x \in [0,1],
        \end{equation}
        subject to homogeneous boundary conditions
        \begin{equation}
            u(0) = u(1) = 0.
        \end{equation}
        Here, \(f : [0,1] \times \mathbb{R} \to \mathbb{R}\) is a given (possibly nonlinear) function. The
        objective is to determine a function \(u^\ast : [0,1] \to \mathbb{R}\) that satisfies the differential
        equation together with the boundary conditions.

        In general, closed-form solutions are not available, and one must resort to numerical approximation.
        In this work, we adopt a collocation-based perspective. Rather than solving the differential
        equation directly, we seek an approximate solution within a parametrized ansatz class of the form
        \(u(x,w)\), where \(w \in W_m\) denotes a finite-dimensional parameter vector.

        To quantify how well a candidate function satisfies the differential equation, we introduce the
        residual function
        \begin{equation}
            r(x,w) := u''(x,w) + f\bigl(x,u(x,w)\bigr),
        \end{equation}
        which
        measures the local deviation from the governing equation at a point \(x\).

        The boundary value problem can then be reformulated as the task of identifying parameters \(w\) such
        that the residual is sufficiently small over the domain while the boundary conditionsare satisfied.
        In practice, this condition is enforced on a finite set of collocation points, leading to a discrete
        approximation of the continuous problem.

        This formulation provides the basis for the quantum collocation method developed in the following
        sections.
        \section{Main Results}
        \label{sec:overview} The main contributions of this work are summarized as
        follows.
        \begin{itemize}
            \item We formulate one-dimensional boundary value problems as residual-based quantum search problems
            over discretized parameter spaces and construct a reversible residual-threshold oracle
            (Definition~\ref{def_final_oracle}) acting jointly on spatial and parameter registers.
            Definition~\ref{def_oracle_score} and Lemma~\ref{lemma:weighted_oracle_expectation} show that the
            resulting oracle induces weighted rather than binary amplification, where the response of a
            parameter configuration \(w_j\) depends on the fraction of collocation points satisfying the
            residual threshold condition.

            \item Theorem~\ref{theorem:coherent_amplitude_amplification} shows that the amplification dynamics
            decompose into a coherent superposition of spatially conditionedamplitude amplification rotations,
            rather than a single global amplification process. This provides a natural extension of the standard
            amplitude amplification perspective tospatially conditioned oracle responses, while preserving
            compatibility with the usual quadratic search mechanism in the parameter space. The resulting
            success probability is governed by the spatial distribution of the local oracle responses.

            \item Theorem~\ref{thm:oracle_cost} shows that, under a reversible implementation of the residual
            evaluation, the elementary gate complexity of the residual oracle scales polynomially in the spatial
            register size \( n_X=\log_2(N_X), \) rather than by an explicit linear enumeration of the \(N_X\)
            collocation points. Thus, one source of efficiency is the coherent evaluation of the residual over
            the spatial register. In the limiting case where the spatial dependence reduces to a standard binary
            marking of the parameterregister, Corollary~\ref{cor:standard_grover_limit} recovers the usual
            Grover scaling.

            \item Numerical experiments validate the theoretical amplification behavior and illustrate the
            influence of discretization, oracle tolerance, finite-precision effects, and ansatz expressivity on
            the residual-based search dynamics.
        \end{itemize}
        \section{Quantum Collocation Method}
        \label{sec:quantum_collocation_method} This section develops
        the quantum implementation of the residual-based collocation framework. We first introduce the
        register encodingsused to represent spatial points, parameters, and residual values. We then
        construct the residual-threshold oracle in compute--phase--uncompute form and analyze the induced
        amplification geometry.

        The construction is organized around the fact that the residual test is evaluated on pairs
        \((x_i,w_j)\). This joint dependence determines the oracle structure and is the source of the
        spatially conditioned amplification behavior derived below.
        \subsection{Quantum Registers and Encodings}
        We use a state register \(X\), a parameter register \(W\), and asigned value register \(Z\). Their
        computational basis states encode, respectively, collocation points \(x_i\in X\), discretized
        parameter vectors \(w_j\in W_m\), and fixed-point values \(z_k\in Z\):
        \[
            \ket{i}_X \leftrightarrow
            x_i, \qquad \ket{j}_W \leftrightarrow w_j, \qquad \ket{k}_Z \leftrightarrow z_k.
        \]
        The signed value
        encoding is used to store residuals and intermediate arithmetic results and, in particular, permits
        operations conditioned on the sign of the encoded value.

        Reversible arithmetic may additionally use an \(n_A\)-qubit ancilla register \(A\). The full
        computational Hilbert space is
        \[
            \mathcal H = \mathcal H_X \otimes \mathcal H_W \otimes \mathcal
            H_Z \otimes \mathcal H_A,
        \]
        where the value and ancilla registers are initialized, unless stated
        otherwise, in the clean state
        \[
            \ket{0}_Z\ket{0}_A.
        \]
        The register dimensions, computational bases, discretized sets, and fixed-point encodings are
        specified in Appendix~\ref{subsec_quantum_registers}. Other numerical representations may be used
        instead, with the usual trade-offs between qubit count, arithmetic depth, and numerical precision;
        see Refs.~\cite{Seidel:2021,Wang:2025}.
        \subsection{Residual Threshold Oracle}
        We consider ansatz functions
        \[
            u:X\times W_m\to Z
        \]
        belonging to a finite-dimensional polynomial
        ansatz space \(V\subseteq\mathcal F\), where \(\mathcal F\subseteq\{f:X\to Z\}\). Throughout this
        work, the ansatz is linear in the parameter vector \(w=(w^0,\ldots,w^{m-1})\in W_m\),
        \[
            u(x,w) =
            \sum_{\ell=0}^{m-1} w^\ell\phi_\ell(x), \qquad \phi_\ell(x) = \sum_{d=0}^{D_\ell}a_{\ell,d}x^d,
        \]
        where the polynomial basis functions \(\phi_\ell\in V\) are chosen to satisfy the prescribed
        boundary conditions. The complexity analysis below only requires that the basis functions admit such
        a polynomial representation. More generally, richer polynomial or trigonometric ansatz spaces may be
        employed. The numerical examples use carefully selected low-degree polynomial ansatz functions in
        order to isolate the effects under investigation while keeping the oracle implementation compact.

        The objective is to identify parameters \(w\in W_m\) for which \(u(\cdot,w)\) approximately
        satisfies the differential equationon the collocation set \(X\). We therefore introduce the residual
        map
        \[
            r:X\times W_m\to Z,
        \]
        where \(r(x_i,w_j)\) measures the local defect of the ansatz at the
        collocation point \(x_i\). For a prescribed tolerance \(\varepsilon_{\mathrm{tol}}\in Z\),
        \(\varepsilon_{\mathrm{tol}}>0\), define the comparison value
        \[
            c(x_i,w_j) = |r(x_i,w_j)| -
            \varepsilon_{\mathrm{tol}}.
        \]
        The residual threshold condition is
        \[
            c(x_i,w_j) < 0.
        \]
        We assume
        that the fixed-point discretization represents all residual, tolerance, and comparison values
        occurring in the computation. Hence, for every \((i,j)\in\mathcal I_X\times\mathcal I_W\), there
        exists an index \(\kappa(i,j)\in\mathcal I_Z\) such that
        \[
            z_{\kappa(i,j)} = c(x_i,w_j).
        \]
        The oracle is constructed by a compute--phase--uncompute procedure. First, the residual is evaluated
        reversibly, its absolute value is computed, and the tolerance is subtracted. This defines a unitary
        compute operator
        \[
            U_c: \mathcal H\to\mathcal H
        \]
        satisfying, on clean value and ancilla
        registers,
        \begin{equation}
            \label{def_computational_operator} U_c \ket{i}_X \ket{j}_W \ket{0}_Z
            \ket{0}_A = \ket{i}_X \ket{j}_W \ket{\kappa(i,j)}_Z \ket{\varphi(i,j)}_A,
        \end{equation}
        where
        \(z_{\kappa(i,j)}=c(x_i,w_j)\) and \(\ket{\varphi(i,j)}_A\) denotes the reversible work state.

        More explicitly, let
        \[
            U_r : \mathcal H_X\otimes\mathcal H_W\otimes\mathcal H_Z \to \mathcal
            H_X\otimes\mathcal H_W\otimes\mathcal H_Z
        \]
        denote the reversible residual-evaluation operator,
        \[
            U_a : \mathcal H_Z\otimes\mathcal H_A \to \mathcal H_Z\otimes\mathcal H_A
        \]
        the reversible
        absolute-value operator, and
        \[
            U_s : \mathcal H_Z \to \mathcal H_Z
        \]
        the operator that subtracts
        the prescribed tolerance. The comparison oracle then decomposes as
        \[
            U_c = (I_X\otimes I_W\otimes
            U_s\otimes I_A) (I_X\otimes I_W\otimes U_a) (U_r\otimes I_A).
        \]
        Thus, \(U_r\) evaluates the
        residual in the value register, \(U_a\) reversibly computes its absolute value using the ancilla
        register, and \(U_s\) subtracts the tolerance from the resulting value.

        Explicit reversible implementations of these oracle subroutines,together with the fixed-point
        arithmetic constructions underlying the complexity analysis, are provided in
        Appendix~\ref{section:reversible_implementations}. The admissible computational subspace is chosen
        so that this coincides with ordinary subtraction of \(\varepsilon_{\mathrm{tol}}\) and no
        wrap-around occurs.

        Define the residual-threshold predicate
        \[
            \chi:\mathcal I_X\times\mathcal I_W\to\{0,1\}, \qquad
            \chi(i,j) = \begin{cases} 1, & |r(x_i,w_j)|<\varepsilon_{\mathrm{tol}}, \\ 0, & \text{otherwise}.
        \end{cases}
    \]
    Let \(S_Z\) be the sign-test phase operator on the value register,
    \[
        S_Z\ket{k}_Z =
        (-1)^{\mathbf 1[z_k<0]} \ket{k}_Z,
    \]
    and set
    \[
        U_{\mathrm{phase}} = I_X\otimes I_W\otimes
        S_Z\otimes I_A.
    \]
    Since the value register uses a signed representation, \(S_Z\) is implemented by
    applying a phase to the sign qubit.

    The induced phase-marking operator is
    \[
        S_\chi = U_c^\dagger U_{\mathrm{phase}} U_c.
    \]
    Using
    \eqref{def_computational_operator}, its action on clean computational states is
    \begin{equation}
        \label{def:selective_phase_operator} S_\chi \bigl( \ket{i}_X \ket{j}_W \ket{0}_Z \ket{0}_A \bigr) =
        (-1)^{\chi(i,j)} \ket{i}_X \ket{j}_W \ket{0}_Z \ket{0}_A .
    \end{equation}
    Thus \(S_\chi\) marks
    pairs \((x_i,w_j)\) according to the localresidual threshold condition.

    Equivalently, on the clean computational subspace
    \[
        \mathcal K = \operatorname{span} \left\{
        \ket{i}_X\ket{j}_W\ket{0}_Z\ket{0}_A : (i,j)\in\mathcal I_X\times\mathcal I_W \right\},
    \]
    one has
    the diagonal representation
    \[
        S_\chi\big|_{\mathcal K} = \sum_{i\in\mathcal I_X} \sum_{j\in\mathcal
        I_W} (-1)^{\chi(i,j)} \ket{i}\bra{i}_X \otimes \ket{j}\bra{j}_W \otimes \ket{0}\bra{0}_Z \otimes
        \ket{0}\bra{0}_A .
    \]
    The search space is the parameter register, whereas the spatial register is used coherently to
    evaluate the residual threshold over the collocation set. Let
    \[
        U_{\mathrm{sup},X}\ket{0}_X =
        \sum_{i\in\mathcal I_X} \alpha_i\ket{i}_X, \qquad \sum_{i\in\mathcal I_X}|\alpha_i|^2=1,
    \]
    be the
    spatial state-preparation unitary. In the uniform case, \( \alpha_i=\frac{1}{\sqrt{N_X}},
    i\in\mathcal I_X. \)
    \begin{definition}[Residual-threshold oracle]
        \label{def_final_oracle} The residual-threshold oracle
        is the unitary operator
        \[
            U_{\mathrm{oracle}} = \left( U_{\mathrm{sup},X}^{\dagger} \otimes
            I_W\otimes I_Z\otimes I_A \right) S_\chi \left( U_{\mathrm{sup},X} \otimes I_W\otimes I_Z\otimes I_A
            \right).
        \]
    \end{definition}
    For a parameter basis state \(\ket{j}_W\) and clean work registers,
    \begin{align}
        \label{eq:final_oracle_action} & U_{\mathrm{oracle}} \bigl( \ket{0}_X \ket{j}_W \ket{0}_Z \ket{0}_A
        \bigr) \nonumber                                                                                    \\ & \quad = \left( U_{\mathrm{sup},X}^{\dagger} \otimes I_W\otimes I_Z\otimes I_A
        \right) \sum_{i\in\mathcal I_X} \alpha_i (-1)^{\chi(i,j)} \ket{i}_X \ket{j}_W \ket{0}_Z \ket{0}_A .
    \end{align}
    Hence the oracle coherently accumulates residual-threshold phaseinformation over the
    spatial superposition and then maps the spatial register back to the reference state.

    The circuit realization of one amplification iteration is shown in Figure~\ref{fig:qcircuit_U_c}.
    The central block implements \(S_\chi=U_c^\dagger U_{\mathrm{phase}}U_c\), while the diffusion
    operator acts only on the parameter register.
    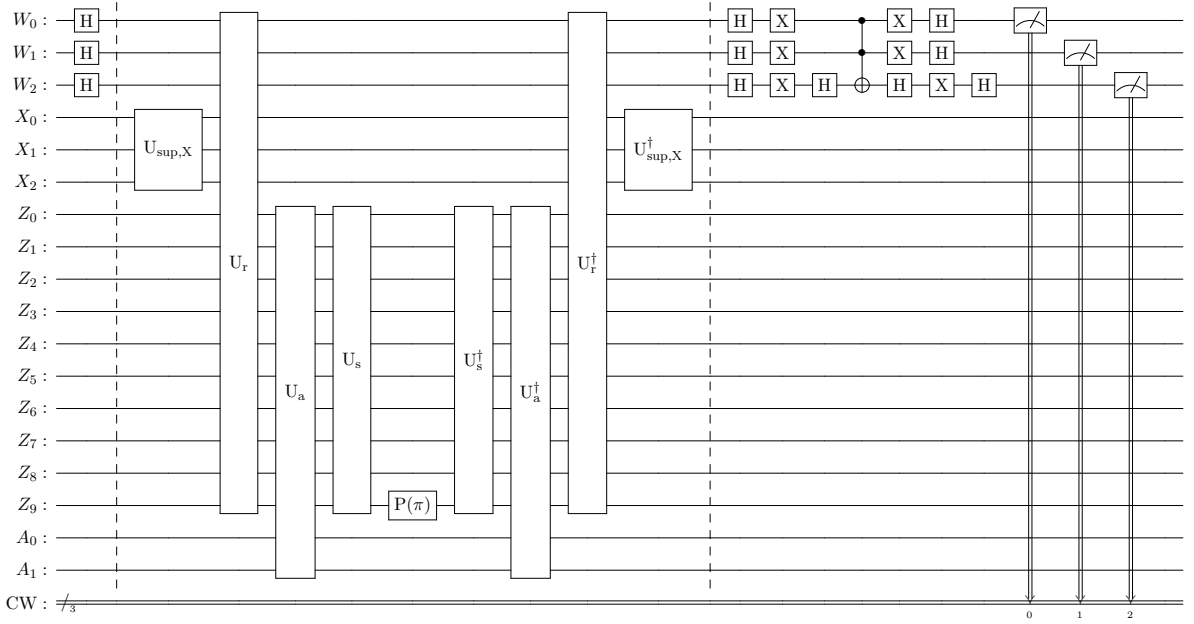
\begin{figure*}[t]
        \centering \resizebox{\textwidth}{!}{
        \Qcircuit @C=1.0em @R=0.2em @!R { \\ \nghost{{W}_{0} : } & \lstick{{W}_{0} : } & \gate{\mathrm{H}}
        \barrier[0em]{17} & \qw & \qw & \multigate{15}{\mathrm{U_r}}_<<<{} & \qw & \qw & \qw & \qw & \qw &
        \multigate{15}{\mathrm{U_r^{\dagger}}}_<<<{} & \qw \barrier[0em]{17} & \qw & \gate{\mathrm{H}} &
        \gate{\mathrm{X}} & \qw & \ctrl{1} & \gate{\mathrm{X}} & \gate{\mathrm{H}} & \qw & \meter & \qw &
        \qw & \qw & \qw\\ \nghost{{W}_{1} : } & \lstick{{W}_{1} : } & \gate{\mathrm{H}} & \qw & \qw &
        \ghost{\mathrm{U_r}}_<<<{} & \qw & \qw & \qw & \qw &\qw & \ghost{\mathrm{U_r^{\dagger}}}_<<<{} & \qw
        & \qw & \gate{\mathrm{H}} & \gate{\mathrm{X}} & \qw & \ctrl{1} & \gate{\mathrm{X}} &
        \gate{\mathrm{H}} & \qw & \qw & \meter & \qw& \qw & \qw\\ \nghost{{W}_{2} : } & \lstick{{W}_{2} : }
        & \gate{\mathrm{H}} & \qw & \qw & \ghost{\mathrm{U_r}}_<<<{} & \qw & \qw & \qw & \qw &\qw &
        \ghost{\mathrm{U_r^{\dagger}}}_<<<{} & \qw & \qw & \gate{\mathrm{H}} & \gate{\mathrm{X}} &
        \gate{\mathrm{H}} & \targ & \gate{\mathrm{H}} & \gate{\mathrm{X}} & \gate{\mathrm{H}} & \qw & \qw &
        \meter & \qw & \qw\\ \nghost{{X}_{0} : } & \lstick{{X}_{0} : } & \qw & \qw &
        \multigate{2}{\mathrm{U_{\mathrm{sup,X}}}}_<<<{} & \ghost{\mathrm{U_r}}_<<<{} & \qw & \qw & \qw &
        \qw & \qw & \ghost{\mathrm{U_r^{\dagger}}}_<<<{} &
        \multigate{2}{\mathrm{U_{\mathrm{sup,X}}^{\dagger}}}_<<<{} & \qw& \qw & \qw & \qw & \qw & \qw & \qw
        & \qw & \qw & \qw & \qw & \qw & \qw\\ \nghost{{X}_{1} : } & \lstick{{X}_{1} : } & \qw & \qw &
        \ghost{\mathrm{U_{\mathrm{sup,X}}}}_<<<{} & \ghost{\mathrm{U_r}}_<<<{} & \qw & \qw & \qw & \qw & \qw
        & \ghost{\mathrm{U_r^{\dagger}}}_<<<{} & \ghost{\mathrm{U_{\mathrm{sup,X}}^{\dagger}}}_<<<{} & \qw &
        \qw & \qw & \qw & \qw & \qw & \qw & \qw & \qw & \qw & \qw & \qw & \qw\\ \nghost{{X}_{2} : } &
        \lstick{{X}_{2} : } & \qw & \qw & \ghost{\mathrm{U_{\mathrm{sup,X}}}}_<<<{} &
        \ghost{\mathrm{U_r}}_<<<{} & \qw & \qw & \qw & \qw & \qw & \ghost{\mathrm{U_r^{\dagger}}}_<<<{} &
        \ghost{\mathrm{U_{\mathrm{sup,X}}^{\dagger}}}_<<<{} & \qw & \qw & \qw & \qw & \qw & \qw & \qw & \qw
        & \qw & \qw & \qw & \qw & \qw\\ \nghost{{Z}_{0} : } & \lstick{{Z}_{0} : } & \qw & \qw & \qw &
        \ghost{\mathrm{U_r}}_<<<{} & \multigate{11}{\mathrm{U_a}}_<<<{} & \multigate{9}{\mathrm{U_s}}_<<<{}
        & \qw & \multigate{9}{\mathrm{U_s^{\dagger}}}_<<<{} & \multigate{11}{\mathrm{U_a^{\dagger}}}_<<<{} &
        \ghost{\mathrm{U_r^{\dagger}}}_<<<{} & \qw & \qw & \qw & \qw & \qw & \qw & \qw & \qw & \qw & \qw &
        \qw & \qw & \qw & \qw\\ \nghost{{Z}_{1} : } & \lstick{{Z}_{1} : } & \qw & \qw & \qw &
        \ghost{\mathrm{U_r}}_<<<{} & \ghost{\mathrm{U_a}}_<<<{} & \ghost{\mathrm{U_s}}_<<<{} & \qw &
        \ghost{\mathrm{U_s^{\dagger}}}_<<<{} & \ghost{\mathrm{U_a^{\dagger}}}_<<<{} &
        \ghost{\mathrm{U_r^{\dagger}}}_<<<{} & \qw & \qw & \qw & \qw & \qw & \qw & \qw & \qw & \qw & \qw &
        \qw & \qw & \qw & \qw\\ \nghost{{Z}_{2} : } & \lstick{{Z}_{2} : } & \qw & \qw & \qw &
        \ghost{\mathrm{U_r}}_<<<{} & \ghost{\mathrm{U_a}}_<<<{} & \ghost{\mathrm{U_s}}_<<<{} & \qw &
        \ghost{\mathrm{U_s^{\dagger}}}_<<<{} & \ghost{\mathrm{U_a^{\dagger}}}_<<<{} &
        \ghost{\mathrm{U_r^{\dagger}}}_<<<{} & \qw & \qw & \qw & \qw & \qw & \qw & \qw & \qw & \qw & \qw &
        \qw & \qw & \qw & \qw\\ \nghost{{Z}_{3} : } & \lstick{{Z}_{3} : } & \qw & \qw & \qw &
        \ghost{\mathrm{U_r}}_<<<{} & \ghost{\mathrm{U_a}}_<<<{} & \ghost{\mathrm{U_s}}_<<<{} & \qw &
        \ghost{\mathrm{U_s^{\dagger}}}_<<<{} & \ghost{\mathrm{U_a^{\dagger}}}_<<<{} &
        \ghost{\mathrm{U_r^{\dagger}}}_<<<{} & \qw & \qw & \qw & \qw & \qw & \qw & \qw & \qw & \qw & \qw &
        \qw & \qw & \qw & \qw\\ \nghost{{Z}_{4} : } & \lstick{{Z}_{4} : } & \qw & \qw & \qw &
        \ghost{\mathrm{U_r}}_<<<<{} & \ghost{\mathrm{U_a}}_<<<{} & \ghost{\mathrm{U_s}}_<<<{} & \qw &
        \ghost{\mathrm{U_s^{\dagger}}}_<<<{} & \ghost{\mathrm{U_a^{\dagger}}}_<<<{} &
        \ghost{\mathrm{U_r^{\dagger}}}_<<<<{} & \qw & \qw & \qw & \qw & \qw & \qw & \qw & \qw & \qw & \qw &
        \qw & \qw & \qw & \qw\\ \nghost{{Z}_{5} : } & \lstick{{Z}_{5} : } & \qw & \qw & \qw &
        \ghost{\mathrm{U_r}}_<<<<{} & \ghost{\mathrm{U_a}}_<<<{} & \ghost{\mathrm{U_s}}_<<<{} & \qw &
        \ghost{\mathrm{U_s^{\dagger}}}_<<<{} & \ghost{\mathrm{U_a^{\dagger}}}_<<<{} &
        \ghost{\mathrm{U_r^{\dagger}}}_<<<<{} & \qw & \qw & \qw & \qw & \qw & \qw & \qw & \qw & \qw & \qw &
        \qw & \qw & \qw & \qw\\ \nghost{{Z}_{6} : } & \lstick{{Z}_{6} : } & \qw & \qw & \qw &
        \ghost{\mathrm{U_r}}_<<<<{} & \ghost{\mathrm{U_a}}_<<<{} & \ghost{\mathrm{U_s}}_<<<{} & \qw &
        \ghost{\mathrm{U_s^{\dagger}}}_<<<{} & \ghost{\mathrm{U_a^{\dagger}}}_<<<{} &
        \ghost{\mathrm{U_r^{\dagger}}}_<<<<{} & \qw & \qw & \qw & \qw & \qw & \qw & \qw & \qw & \qw & \qw &
        \qw & \qw & \qw & \qw\\ \nghost{{Z}_{7} : } & \lstick{{Z}_{7} : } & \qw & \qw & \qw &
        \ghost{\mathrm{U_r}}_<<<<{} & \ghost{\mathrm{U_a}}_<<<{} & \ghost{\mathrm{U_s}}_<<<{} & \qw &
        \ghost{\mathrm{U_s^{\dagger}}}_<<<{} & \ghost{\mathrm{U_a^{\dagger}}}_<<<{} &
        \ghost{\mathrm{U_r^{\dagger}}}_<<<<{} & \qw & \qw & \qw & \qw & \qw & \qw & \qw & \qw & \qw & \qw &
        \qw & \qw & \qw & \qw\\ \nghost{{Z}_{8} : } & \lstick{{Z}_{8} : } & \qw & \qw & \qw &
        \ghost{\mathrm{U_r}}_<<<<{} & \ghost{\mathrm{U_a}}_<<<{} & \ghost{\mathrm{U_s}}_<<<{} & \qw &
        \ghost{\mathrm{U_s^{\dagger}}}_<<<{} & \ghost{\mathrm{U_a^{\dagger}}}_<<<{} &
        \ghost{\mathrm{U_r^{\dagger}}}_<<<<{} & \qw & \qw & \qw & \qw & \qw & \qw & \qw & \qw & \qw & \qw &
        \qw & \qw & \qw & \qw\\ \nghost{{Z}_{9} : } & \lstick{{Z}_{9} : } & \qw & \qw & \qw &
        \ghost{\mathrm{U_r}}_<<<<{} & \ghost{\mathrm{U_a}}_<<<{} & \ghost{\mathrm{U_s}}_<<<{} &
        \gate{\mathrm{P(\pi)}} & \ghost{\mathrm{U_s^{\dagger}}}_<<<{} & \ghost{\mathrm{U_a^{\dagger}}}_<<<{}
        & \ghost{\mathrm{U_r^{\dagger}}}_<<<<{} & \qw & \qw & \qw & \qw & \qw & \qw & \qw & \qw & \qw & \qw
        & \qw & \qw & \qw & \qw\\ \nghost{{A}_{0} : } & \lstick{{A}_{0} : } & \qw & \qw & \qw & \qw &
        \ghost{\mathrm{U_a}}_<<<<{} & \qw & \qw & \qw & \ghost{\mathrm{U_a^{\dagger}}}_<<<<{} & \qw & \qw &
        \qw & \qw & \qw & \qw & \qw & \qw & \qw & \qw & \qw & \qw & \qw & \qw & \qw\\ \nghost{{A}_{1} : } &
        \lstick{{A}_{1} : } & \qw & \qw & \qw & \qw & \ghost{\mathrm{U_a}}_<<<<{} & \qw & \qw & \qw &
        \ghost{\mathrm{U_a^{\dagger}}}_<<<<{} & \qw & \qw & \qw & \qw & \qw & \qw & \qw & \qw & \qw & \qw &
        \qw & \qw & \qw & \qw & \qw\\ \nghost{\mathrm{{CW} : }} & \lstick{\mathrm{{CW} : }} &
        \lstick{/_{_{3}}} \cw & \cw & \cw & \cw & \cw & \cw & \cw & \cw & \cw & \cw & \cw & \cw & \cw & \cw
        & \cw & \cw & \cw & \cw & \cw & \dstick{_{_{\hspace{0.0em}0}}} \cw \ar @{<=} [-18,0] &
        \dstick{_{_{\hspace{0.0em}1}}} \cw \ar @{<=} [-17,0] & \dstick{_{_{\hspace{0.0em}2}}} \cw \ar @{<=}
        [-16,0] & \cw & \cw\\ \\ }} \caption{Circuit schematic of one amplification iteration for the
        residual-based search procedure. The dashed barriers separate three logical regions. On the left,the
        parameter register \(W\) is prepared in a uniform superposition by Hadamard gates, while the spatial
        register \(X\) is prepared by \(U_{\mathrm{sup},X}\) in a superposition over all grid points. The search
        is performed exclusively over \(W\), whereas \(X\) servesas an auxiliary register to enable the coherent
        evaluation of the residual across the spatial domain. In the middle, the oracle is implemented in
        compute--phase--uncompute form. The compute block \(U_c = U_s U_a U_r\) acts on $\mathcal H_X \otimes
        \mathcal H_W \otimes \mathcal H_Z \otimes \mathcal H_A$. Starting from a value register initialized
        in \(\ket{0}_Z\), the circuit coherently evaluates the residuals \(r(x_i,w_j)\) for superpositions of
        \(\ket{i}_X \otimes \ket{j}_W\), computes the comparison value
        \(c(x_i,w_j)=|r(x_i,w_j)|-\varepsilon_{\mathrm{tol}}\), and writesthe discretized result
        \(z_{\kappa(i,j)}\) into the value register \(Z\), while the ancilla register \(A\) stores intermediate
        information required for reversible arithmetic. The phase test is implemented by applying the phase
        gate \(P(\pi)\) to the sign qubit of \(Z\), marking negative values of \(z_{\kappa(i,j)}\).The computation
        is subsequently uncomputed by applying \(U_c^\dagger\). On the right, the Grover diffusion operator
        \(U_{\Psi_W}\) acts only on the parameter register \(W\).} \label{fig:qcircuit_U_c}
    \end{figure*}
    \subsection{Oracle Structure and Amplification Geometry}
    Unlike the standard amplitude amplification framework introducedin Grover's algorithm
    \cite{Grover:1996} and generalized in \cite{Brassard:2000}, the residual-threshold oracle does not
    induce a binary phase marking on the parameter register alone. Instead, the phase response depends
    jointly on the spatial and parameter indices through the predicate \(\chi(i,j)\). On the clean
    computational subspace\(\mathcal K\), the effective oracle therefore acts on \(\mathcal
    H_X\otimes\mathcal H_W\) by assigning phases to pairs\((x_i,w_j)\), rather than to parameter values
    \(w_j\) alone. This joint structure is responsible for the spatially conditioned amplification
    dynamics derived below.
    \subsubsection{Conditioning on \(w_j\)}
    Fix \(j\in\mathcal I_W\). The corresponding conditioned operatoron the spatial register is
    \begin{equation}
        \label{def:S_j} S_\chi^{(j)} = \sum_{i=0}^{N_X-1} (-1)^{\chi(i,j)} \ket{i}\bra{i}_X
        .
    \end{equation}
    Thus \(S_\chi^{(j)}\) is diagonal in the spatial basis and each collocation point
    acquires a phase according to the residual-threshold test for the pair \((x_i,w_j)\).

    Let
    \[
        \ket{\Psi_X} = \sum_{i=0}^{N_X-1}\alpha_i\ket{i}_X, \qquad \sum_{i=0}^{N_X-1}|\alpha_i|^2=1,
    \]
    denote the spatial state prepared by
    \[
        U_{\mathrm{sup},X}\ket{0}_X=\ket{\Psi_X}.
    \]
    \begin{definition}[Effective oracle score]
        \label{def_oracle_score} The quantity \(q(w_j)\) is the
        spatially weighted fraction of collocation points at which the parameter value \(w_j\) satisfies the
        residual threshold. For \(j\in\mathcal I_W\), define
        \begin{align}
            q(w_j) = \sum_{i=0}^{N_X-1}
            |\alpha_i|^2\chi(i,j).
        \end{align}

    \end{definition}
    \begin{lemma}[Weighted oracle expectation]
        \label{lemma:weighted_oracle_expectation} For every
        \(j\in\mathcal I_W\),
        \[
            \bra{\Psi_X}S_\chi^{(j)}\ket{\Psi_X} = 1-2q(w_j).
        \]
    \end{lemma}
    \begin{proof}
        Using \eqref{def:S_j},
        \[
            \bra{\Psi_X}S_\chi^{(j)}\ket{\Psi_X} = \sum_{i=0}^{N_X-1}
            |\alpha_i|^2(-1)^{\chi(i,j)}.
        \]
        Since \((-1)^{\chi(i,j)}=1-2\chi(i,j)\), the claim follows directly
        from the definition of \(q(w_j)\).
    \end{proof}
    Equivalently, if
    \[
        \Pi_{\mathrm{good}}^{(j)} = \sum_{\chi(i,j)=1} \ket{i}\bra{i}_X, \qquad
        S_\chi^{(j)} = I_X-2\Pi_{\mathrm{good}}^{(j)},
    \]
    then
    \[
        q(w_j) = \bra{\Psi_X}
        \Pi_{\mathrm{good}}^{(j)} \ket{\Psi_X}.
    \]
    Hence \(q(w_j)\) determines the effective oracle response
    associated with \(w_j\). In general, this response is not binary on the parameter space and
    therefore does not correspond to a single global Grover rotation angle.
    \subsubsection{Conditioning on \(x_i\): coherent amplitude amplification on \(W\)}
    We now condition on a fixed spatial index \(i\). Since the effective oracle is diagonal in the
    spatial register, each basis state \(\ket{i}_X\) induces a phase oracle on the parameter register,
    \[
        S_\chi^{(i)} = \sum_{j=0}^{N_W-1} (-1)^{\chi(i,j)} \ket{j}\bra{j}_W .
    \]
    Equivalently,
    \[
        S_\chi^{(i)} = I_W-2\Pi_{\mathrm{good}}^{(i)}, \qquad \Pi_{\mathrm{good}}^{(i)} = \sum_{j\in\mathcal
        M_i} \ket{j}\bra{j}_W,
    \]
    where
    \[
        \mathcal M_i = \{\,j\in\mathcal I_W:\chi(i,j)=1\,\}.
    \]
    Thus, for
    fixed \(x_i\), the oracle marks precisely those parameter values whose residual satisfies the
    threshold at that collocation point.

    Let
    \begin{equation}
        \label{def_grover_diffusion} U_{\Psi_W} = 2\ket{\Psi}\bra{\Psi}_W-I_W
    \end{equation}
    be the diffusion operator about the uniform parameter state
    \[
        \ket{\Psi}_W =
        \frac{1}{\sqrt{N_W}} \sum_{j=0}^{N_W-1}\ket{j}_W .
    \]
    The amplification iterate is
    \[
        Q = \bigl(
        I_X\otimes U_{\Psi_W}\otimes I_Z\otimes I_A \bigr) S_\chi .
    \]
    For fixed \(i\), define
    \[
        Q_i =
        U_{\Psi_W}S_\chi^{(i)} .
    \]
    On the clean computational subspace, the full iterate decomposesas
    \begin{equation}
        \label{grover_decomposition} Q\big|_{\mathcal K} = \sum_{i=0}^{N_X-1}
        \ket{i}\bra{i}_X \otimes Q_i \otimes \ket{0}\bra{0}_Z \otimes \ket{0}\bra{0}_A .
    \end{equation}
    Thus
    the global evolution is a coherently controlled family of amplitude amplification processes on
    \(\mathcal H_W\), indexed by the spatial register.

    For each \(i\), define the normalized good and bad states
    \[
        \ket{g_i} = \frac{1}{\sqrt{|\mathcal
        M_i|}} \sum_{j\in\mathcal M_i}\ket{j}_W, \qquad \ket{b_i} = \frac{1}{\sqrt{N_W-|\mathcal M_i|}}
        \sum_{j\notin\mathcal M_i}\ket{j}_W,
    \]
    and set
    \[
        \mathcal G_i =
        \operatorname{span}\{\ket{g_i},\ket{b_i}\}.
    \]
    The uniform parameter state satisfies
    \[
        \ket{\Psi}_W
        = \sin\theta_i\,\ket{g_i} + \cos\theta_i\,\ket{b_i}, \qquad \sin^2\theta_i = \frac{|\mathcal
        M_i|}{N_W}.
    \]
    Both \(S_\chi^{(i)}\) and \(U_{\Psi_W}\) leave \(\mathcal G_i\) invariant.
    Consequently, \(Q_i\) acts on \(\mathcal G_i\) as the standard two-dimensional amplitude
    amplification rotation \cite{Grover:1996,Brassard:2000}. With respect to the ordered basis
    \((\ket{g_i},\ket{b_i})\),
    \[
        Q_i = \begin{pmatrix} \cos(2\theta_i)  & \sin(2\theta_i) \\
        -\sin(2\theta_i) & \cos(2\theta_i)\end{pmatrix}.
    \]
    Therefore,
    \[
        Q_i^k\ket{\Psi}_W =
        \sin((2k+1)\theta_i)\ket{g_i} + \cos((2k+1)\theta_i)\ket{b_i}.
    \]
    For a spatial superposition
    \[
        \ket{\Psi_X} = \sum_{i=0}^{N_X-1}\alpha_i\ket{i}_X,
    \]
    the block
    decomposition \eqref{grover_decomposition} gives
    \[
        Q^k \bigl( \ket{\Psi_X} \otimes \ket{\Psi}_W
        \otimes \ket{0}_Z \otimes \ket{0}_A \bigr) = \sum_{i=0}^{N_X-1} \alpha_i \ket{i}_X \otimes
        Q_i^k\ket{\Psi}_W \otimes \ket{0}_Z \otimes \ket{0}_A .
    \]
    Thus the amplification process evolves as
    a coherent superposition of spatially conditioned amplitude amplification rotations, generally with
    different angles \(\theta_i\).
    \begin{theorem}[Coherent spatially conditioned amplitude amplification]
        \label{theorem:coherent_amplitude_amplification} With the notation above,
        \[
            Q^k \bigl( \ket{\Psi_X}
            \otimes \ket{\Psi}_W \otimes \ket{0}_Z \otimes \ket{0}_A \bigr) = \sum_{i=0}^{N_X-1} \alpha_i
            \ket{i}_X \otimes Q_i^k\ket{\Psi}_W \otimes \ket{0}_Z \otimes \ket{0}_A .
        \]
        The total success
        probability after \(k\) iterations is
        \[
            P(k) = \sum_{i=0}^{N_X-1} |\alpha_i|^2
            \sin^2((2k+1)\theta_i),
        \]
        where
        \[
            \theta_i = \arcsin \left( \sqrt{\frac{|\mathcal M_i|}{N_W}}
            \right), \qquad \mathcal M_i = \{\,j\in\mathcal I_W:\chi(i,j)=1\,\}.
        \]
        Consequently, the optimal
        iteration count is characterized by
        \[
            k^\ast = \arg\max_{k\in\mathbb N_0} \sum_{i=0}^{N_X-1}
            |\alpha_i|^2 \sin^2((2k+1)\theta_i).
        \]
    \end{theorem}
    \begin{proof}
        Equation~\eqref{grover_decomposition} implies that, on the cleancomputational
        subspace,
        \[
            Q^k = \sum_{i=0}^{N_X-1} \ket{i}\bra{i}_X \otimes Q_i^k \otimes \ket{0}\bra{0}_Z
            \otimes \ket{0}\bra{0}_A .
        \]
        Applying this identity to \(
        \ket{\Psi_X}\otimes\ket{\Psi_W}\otimes\ket{0}_Z\otimes\ket{0}_A \) gives the stated coherent
        decomposition.

        For each spatial state \(i\), the operator \(Q_i\) is the standard amplitude amplification iterate
        on \(\mathcal G_i\). Hence
        \[
            Q_i^k\ket{\Psi}_W = \sin((2k+1)\theta_i)\ket{g_i} +
            \cos((2k+1)\theta_i)\ket{b_i}.
        \]
        The probability of measuring a marked parameter in sector \(i\) is
        therefore \(\sin^2((2k+1)\theta_i)\). Since the spatial basis states are orthogonal, tracing over
        the spatial register gives
        \[
            P(k) = \sum_{i=0}^{N_X-1} |\alpha_i|^2 \sin^2((2k+1)\theta_i).
        \]
        Maximization over \(k\in\mathbb N_0\) yields the stated expression for \(k^\ast\).
    \end{proof}
    \begin{remark}[Spatially coherent amplification geometry]
        In practical boundary value problems, the
        marked parameter set
        \[
            \mathcal M_i = \{\,j\in\mathcal I_W:\chi(i,j)=1\,\}
        \]
        usually depends on
        the spatial index \(i\). Consequently, the algorithm is not governed by a single global
        amplification angle. Instead,it realizes a coherent superposition of spatially conditioned amplitude
        amplification rotations, with weights determined by the spatial amplitudes \(\alpha_i\).
    \end{remark}
    \section{Oracle Cost
    and Runtime Comparison}
    \label{sec:oracle_cost_runtime} We now separate the cost of constructing the
    residual-threshold oracle from the subsequent amplitude amplification dynamics. The main point is
    that the boundary value problem residual is evaluated coherentlyon the spatial register, so that the
    oracle cost depends polynomially on \(n_X=\log_2(N_X)\), rather than linearly on the number \(N_X\)
    of collocation points.
    \begin{theorem}[Cost of the residual-threshold oracle]
        \label{thm:oracle_cost} Assume that the
        ansatz admits the polynomial representation
        \[
            u(x,w) = \sum_{j=0}^{m-1} \sum_{d=0}^{D_j}
            a_{j,d}\,w_j x^d,
        \]
        and define
        \[
            D_u := \max_{0\le j\le m-1} D_j .
        \]
        Furthermore, assume that the
        forcing term is polynomial in \(u\),
        \[
            f(x,u) = \sum_{q=0}^{Q} b_q(x)u^q, \qquad
            \deg_x\!\bigl(b_q(x)\bigr)=B_q .
        \]
        Define the maximal induced spatial degree
        \[
            D_f := \max_{0\le
            q\le Q} \bigl( B_q+qD_u \bigr).
        \]
        Using the reversible QFT-based arithmetic constructions of Appendix~\ref{subsec:qft_polynomial}, the
        residual-threshold oracle from Definition~\ref{def_final_oracle} can be implemented with elementary
        gate complexity
        \[
            C_{\mathrm{oracle}} = O\!\left( n_Z n_X^{D_f} n_W^{Q}(D_f+Q+1)^2 + n_Z^2 \right),
        \]
        up to lower-order polynomial overhead from finite-difference shifts, comparison, phase marking,
        and uncomputation.
    \end{theorem}
    \begin{proof}
        The residual evaluation consists of the finite-difference contribution \(D_h^2u(x,w)\)
        and the forcing contribution \(f(x,u(x,w))\).

        Since \(u(x,w)\) is linear in the parameter variables and polynomial in \(x\) of degree at most
        \(D_u\), every monomial appearing in \(D_h^2u(x,w)\) has parameter degree at most \(1\) and spatial
        degree at most \(D_u\). By Lemma~\ref{lemma:qft_polynomial_runtime}, each such monomial can be
        implemented with gate complexity
        \[
            O\!\left( n_Z n_X^{D_u} n_W (D_u+2)^2 + n_Z^2 \right).
        \]
        We now consider the forcing contribution
        \begin{align*}
            f(x,u(x,w)) = \sum_{q=0}^{Q} b_q(x)u(x,w)^q
            .
        \end{align*}
        For fixed \(q\), every monomial generated by \(b_q(x)u(x,w)^q\) has parameter degree
        at most \(q\) and spatial degree at most \( B_q+qD_u . \) Therefore
        Lemma~\ref{lemma:qft_polynomial_runtime} gives, for the \(q\)-th summand, the gate-complexity bound
        \[
            O\!\left( n_Z n_X^{B_q+qD_u} n_W^{q} \bigl(B_q+qD_u+q+1\bigr)^2 + n_Z^2 \right).
        \]
        By the definitions of \(D_f\) and \(Q\), we have
        \[
            B_q+qD_u\le D_f, \qquad q\le Q,
        \]
        for all
        \(q\in\{0,\ldots,Q\}\). Hence each forcing summand is bounded by
        \[
            O\!\left( n_Z n_X^{D_f} n_W^{Q}
            (D_f+Q+1)^2 + n_Z^2 \right).
        \]
        Summing over \(q=0,\ldots,Q\) changes the estimate only by the
        polynomial factor \(Q+1\), which is absorbed into the stated asymptotic bound.

        The remaining oracle components, namely finite-difference shifts, absolute-value evaluation,
        threshold comparison, phase marking, and the corresponding uncomputation, contribute only
        lower-order polynomial overhead under the same arithmetic model. Therefore
        \[
            C_{\mathrm{oracle}} =
            O\!\left( n_Z n_X^{D_f} n_W^{Q} (D_f+Q+1)^2 + n_Z^2 \right).
        \]
    \end{proof}
    \begin{remark}
        The complexity estimate of Theorem~\ref{thm:oracle_cost} should be interpreted as a
        worst-case upper bound. It is derived from the explicit QFT-based polynomial arithmetic construction
        of Section~\ref{subsec:qft_polynomial}, where polynomial terms are expanded into bit-level monomial
        contributions and implemented through multi-controlled phase operations. Consequently, the estimate
        reflects the maximal combinatorial growth induced by the polynomial degrees \(D_f\) and \(Q\). For more
        specialized reversible arithmetic implementations, lower gate complexities may be achievable.
    \end{remark}
    \begin{remark}[Comparison with classical exhaustive collocation search]
        A classical exhaustive collocation search evaluates the residualindependently for all pairs
        \[
            (x_i,w_j)\in X\times W_m .
        \]
        Let \( C_r^{\mathrm{cl}}(n_Z) \) denote the number of classical
        arithmetic operations required toevaluate the fixed-point residual \( r(x_i,w_j) \) at precision
        \(n_Z\). The total classical computational cost is therefore
        \[
            C_{\mathrm{classical}} = O\!\left(
            N_X N_W C_r^{\mathrm{cl}}(n_Z) \right).
        \]
        In contrast, the quantum implementation encodes the spatial and parameter discretizations using
        \[
            n_X=\log_2(N_X), \qquad n_W=\log_2(N_W),
        \]
        qubits and evaluates the residual coherently across the superposition of collocation points and
        parameter states. By Theorem~\ref{thm:oracle_cost}, the resulting oracle complexity scales as
        \[
            C_{\mathrm{oracle}} = O\!\left( n_Z n_X^{D_f} n_W^Q \right).
        \]
        Hence, under the polynomial arithmetic model of Appendix~\ref{subsec:qft_polynomial}, the oracle
        complexity is polynomial in the logarithmic register sizes \( n_X \) and \( n_W \), rather than
        linear in the explicit discretization sizes \( N_X \) and \( N_W \).
    \end{remark}
    \begin{corollary}[Reduction to the standard Grover regime]
        \label{cor:standard_grover_limit}
        Assume that the residual-threshold predicate is independent of the spatial index. That is, suppose
        there exists a set
        \[
            \mathcal M \subseteq \mathcal I_W
        \]
        such that
        \[
            \chi(i,j) = \begin{cases}1,
            & j\in\mathcal M, \\ 0, & j\notin\mathcal M,\end{cases} \qquad \text{for all } i\in\mathcal I_X .
        \]
        Let
        \[
            M:=|\mathcal M|, \qquad 0<M<N_W .
        \]
        Then the marked parameter set is identical for all collocation points, and all amplification angles
        coincide:
        \begin{align}
            \label{def_superposition_theta} \theta_i = \theta = \arcsin\!\left(
            \sqrt{\frac{M}{N_W}} \right).
        \end{align}
        Consequently, the coherent spatially conditioned
        amplification dynamics reduce to the standard amplitude amplification regime on \(\mathcal H_W\).
        The optimal iteration count is given by
        \begin{align}
            \label{def_optimal_k} k^\ast = \left\lfloor
            \frac{\pi}{4\theta} -\frac12 \right\rfloor,
        \end{align}
        where \(\lfloor\cdot\rfloor\) denotes
        rounding to the nearest integer.

        The corresponding total gate complexity is
        \[
            C_{\mathrm{quantum}} = O\!\left( \frac{1}{
            \arcsin\!\left( \sqrt{M/N_W} \right) } \left( n_Z n_X^{D_f} n_W^Q (D_f+Q+1)^2 + n_Z^2 \right)
            \right).
        \]
        In the sparse-marked regime \( \frac{M}{N_W}\ll1, \) one has \( \arcsin\!\left( \sqrt{\frac{M}{N_W}}
        \right) \sim \sqrt{\frac{M}{N_W}}, \) and therefore
        \[
            C_{\mathrm{quantum}} = O\!\left(
            \sqrt{\frac{N_W}{M}} \left( n_Z n_X^{D_f} n_W^Q (D_f+Q+1)^2 + n_Z^2 \right) \right).
        \]
    \end{corollary}
    \begin{proof}
        Under the stated assumption, \( \mathcal M_i = \mathcal M \text{ for all } i\in\mathcal I_X , \) and
        therefore all amplification angles coincide:
        \[
            \theta_i = \arcsin\!\left( \sqrt{\frac{M}{N_W}}
            \right) =: \theta .
        \]
        Since \( \sum_{i=0}^{N_X-1} |\alpha_i|^2=1. \) the success probability from
        Theorem~\ref{theorem:coherent_amplitude_amplification} reduces to
        \[
            P(k) = \sum_{i=0}^{N_X-1}
            |\alpha_i|^2 \sin^2((2k+1)\theta) = \sin^2((2k+1)\theta),
        \]
        Thus the spatially conditioned dynamics reduce to standard amplitude amplification on \(\mathcal
        H_W\). The success probability is maximized when
        \[
            (2k+1)\theta \approx \frac{\pi}{2},
        \]
        which
        yields
        \[
            k^\ast = \left\lfloor \frac{\pi}{4\theta} -\frac12 \right\rfloor .
        \]
        Substituting \( \theta = \arcsin\!\left( \sqrt{\frac{M}{N_W}} \right) \) into
        Theorem~\ref{thm:oracle_cost} gives the stated complexity estimate.
    \end{proof}
    \begin{remark}[Uniform spatial weighting]
        If the spatial preparation is uniform, i.e.
        \[
            U_{\mathrm{sup},X}\ket{0}_X = \frac{1}{\sqrt{N_X}} \sum_{i=0}^{N_X-1}\ket{i}_X,
        \]
        then the coherent
        success probability simplifies to
        \[
            P(k) = \frac{1}{N_X} \sum_{i=0}^{N_X-1} \sin^2((2k+1)\theta_i).
        \]
        This is not, in general, the standard Grover setting, since the amplification angles \(\theta_i\)
        may still depend on the spatial index. The standard Grover limit is recovered only when the marked
        parameter set is independent of \(i\).
    \end{remark}
    \section{Experimental Setup}
    \label{sec:experimental_setup} The aim of the experiments is to study
    the interplay between discretization error, ansatz error, and finite-precision effectsin a
    controlled setting. In particular, we analyze how the performance of the method depends on the
    Grover iteration count \(k\), the spatial resolution \(n_X\), and the oracle tolerance \(\varepsilon\)
    (used as shorthandfor \(\varepsilon_{\mathrm{tol}}\) throughout the following analysis), in relation
    to the register precisions \((n_Z, n_W)\).
    \subsection{Problem Class}
    We consider one-dimensional boundary value problems
    \begin{equation}
        u''(x)+f(x,u(x))=0,\qquad
        x\in[0,1], \qquad u(0)=u(1)=0.
    \end{equation}
    The numerical experiments use three representative test cases that are simple enough to be tested on a simulator:
    \begin{align}
        & u''(x)+1=0, \label{eq:test_constant} \\ & u''(x)+x=0,
        \label{eq:test_linear}                  \\ & u''(x)+u(x)+\alpha u(x)^2 +2\lambda-\lambda x(1-x) -\alpha\lambda^2
        x^2(1-x)^2=0, \qquad \alpha\in\{8,16,32\}. \label{eq:test_nonlinear}
    \end{align}
    The first two
    problems are linear test cases with constant and spatially varying forcing, respectively, while the
    third is a manufacturednonlinear problem in which \(\alpha\) controls the strength of the
    nonlinearity.
    \subsection{Ansatz Functions}
    We use parametrized ansatz functions of increasing expressivity:
    \begin{align}
        u_{1}(x) & = w_2\,
        x(1-x), \label{eq:u1} \\ u_{2}(x) & = w_2\, x(1-x) + w_1\, x^2 (1-x)^2, \label{eq:u2} \\ u_{3}(x) &
        = w_2\, x(1-x^2). \label{eq:u3}
    \end{align}
    The second ansatz defines a two-dimensional parameter
    space, while the first and third corresponds to a one-dimensional parameter space.
    \subsection{Parameter Study}
    The parameter study is designed to test the theoretical predictions
    under controlled changes of the discretization, oracle tolerance, ansatz class, finite precision,
    and oracle implementation. We organize the experiments as follows.
    \begin{enumerate}
        \item
        \textbf{Linear baseline problem.} We first consider a linear problem whose exact solution is
        contained in the ansatz class. Both one- and two-dimensional parameterizations are used to verify
        the amplification dynamics and the extension to higherdimensional parameter spaces. For a unique
        admissible solution, the observed optimal Grover iteration count is compared with
        Corollary~\ref{cor:standard_grover_limit}. Increasing the tolerance \(\varepsilon\) produces
        multiple admissible solutionsand is used to test the corresponding generalized Grover behavior.

        \item \textbf{Linear inhomogeneous problem.} We then consider a problem with spatially varying
        forcing. By varying \(n_X\), we compare a case in which the exact solution is representable by the
        ansatz with a case in which it is not. This separates discretization error from approximation error
        and shows how the two interact. In the same setting, we also use representative tolerance values to
        check the amplification dynamics predicted by
        Theorem~\ref{theorem:coherent_amplitude_amplification}. This serves as a consistency test for the
        weighted-oracle mechanism, including the dependence of the effective Grover angles and sampling
        distributions on the induced marked parameter sets.

        \item \textbf{Nonlinear manufactured problem.} We next study a nonlinear manufactured problem for
        which the ansatz class contains an exact solution. The parameters \(\varepsilon\), \(n_Z\), and
        \(\alpha\) are varied to assess the effects of tolerance, finite precision, and nonlinearity. The
        focus is the induced residual landscape: larger \(\alpha\) may create severallow-residual regions in
        parameter space, which can be amplified even when they are not true solutions of the boundary value
        problem. \item \textbf{Oracle resource scaling.} Finally, we evaluate thereversible oracle
        implementation. After decomposing the circuits into a native gate basis, we record gate counts,
        circuit depth, and qubit number as functions of \(n_X\), and compare the observed scaling with
        Theorem~\ref{thm:oracle_cost}.
    \end{enumerate}
    Across all experiments, we employ discretizations of
    the spatialdiscretization parameter \(n_X\), the parameter space resolution \(n_W\), and the value
    register precision \(n_Z\), and systematically vary the oracle tolerance \(\varepsilon\). While the
    number of Grover iterations \(k\) is explicitly varied in the baseline experiment to study amplitude
    amplification, theremaining experiments primarily rely on oracle-based landscape exploration
    combined with fixed-depth Grover steps, enabling a consistent analysis of residual structures under
    varying discretization and tolerance parameters.
    \subsection{Error Decomposition}
    The discretized and parametrized formulation introduces three main sources of approximation error.
    \paragraph{Ansatz approximation error}
    Let \(u^\ast\) denote the exact solution and \(\mathcal
    U=\{u(\cdot,w)\mid w\in W_m\}\) the ansatz class. We define
    \[
        e_{\mathrm{approx}} := \inf_{w\in
        W_m} \|u^\ast-u(\cdot,w)\|_X, \qquad \|v\|_X:=\max_{x_i\in X}|v(x_i)|,
    \]
    where
    \(X=\{x_i\}_{i=0}^{N_X-1}\) is the collocation grid. In general, \(e_{\mathrm{approx}}>0\) unless
    \(u^\ast\in\mathcal U\). Thus, even a vanishing discrete residual does not by itself imply agreement
    with the exact solution on the continuous domain.
    \paragraph{Spatial discretization error}
    The differential operator is evaluated on the finite
    collocationgrid. For the second derivative we use the centered finite-difference approximation
    \[
        D_h^2u(x) = \frac{u(x-h)-2u(x)+u(x+h)}{h^2} = u''(x)+\mathcal O(h^2),
    \]
    for sufficiently smooth
    \(u\) \cite{Leveque:2007}.
    \paragraph{Finite-precision quantization error}
    Residual values and intermediate quantities are
    represented by fixed-point registers. A value-register precision \(p_Z\) induces a quantization
    scale \(2^{-p_Z}\), leading to rounding, truncation, and possible accumulation errors in the
    reversible arithmetic circuit.

    The overall approximation quality is governed by the interactionof these three effects.
    \section{Numerical Experiments}
    \label{sec:experimental_analysis}
    All simulations were performed using Qiskit (version 2.3.0) withthe Aer simulator (version 0.17.2).
    Each quantum circuit was executed with 1000 measurement shots. The spatial and parameter registers
    use different fixed-point conventions. The spatial register encodes unsigned coordinates on
    \([0,1]\). Since the finite-difference residual evaluation uses shifted arguments such as \(x+h\),
    we choose \(n_X=p_X+1\), where the additional qubit provides overflow capacity for intermediate
    coordinate arithmetic without enlarging the physical collocation domain. The parameter register
    encodes the search space by a signed fixed-point representation: for each parameter, \(n_W=p_W+1\),
    with one sign bit and \(p_W\) fractional bits, yielding a uniform grid on\([-1,1)\) with spacing
    \(2^{-p_W}\). The value-register size \(n_Z\) is chosen from a heuristic boundon the residual range,
    \begin{equation}
        n_Z \approx 1+p_Z+2p_X+r_{\mathrm{res}}. \label{nZguess}
    \end{equation}
    Here
    \(p_Z\) denotes the fixed-point precision of the encoded values, while \(2p_X\) reflects the
    \(h^{-2}\) scaling of the centered second-difference operator. The leading \(1\) accounts for the
    sign bit, and \(r_{\mathrm{res}}\) is an overflow margin, taken as \(r_{\mathrm{res}}=1\) unless
    stated otherwise.
    \begin{table}[!htbp]
        \centering \caption{Overview of the numerical experiments.} \begin{tabular}{c c
        c c c c c c} \toprule Exp. & Focus                  & Ansatz            & \(n_W\) & \(n_X\)     & \(n_Z\)       & \(\varepsilon\)       & \(k\) \\
        \midrule

        1                              & Amplification dynamics & \(u_1,\;u_2\)       & \(6,8\) & \(3\)       & \(9\)         & \(2^{-1}\!-\!2^{-4}\) & \(1\!-\!6\)                           \\
        2                              & Spatial resolution     & \(u_3,\;u_1\)       & \(4,5\) & \(3\!-\!5\) & \(10\!-\!14\) & \(2^{-1}\!-\!2^{-3}\) &
        \(1\!-\!3\)                                                                                                                                                                 \\

        3                              & Nonlinearity           & \(u_1\)             & \(6\)   & \(3\)       & \(8\!-\!10\)  & \(2^{-1}\!-\!2^{-4}\) & \(1\)                                 \\
        4                              & Scaling                & \(u_1,\;u_3,\;u_1\) & \(4\)   & \(3\!-\!6\) & \(6\)         & --                  & --                                  \\
        \bottomrule\end{tabular} \label{tab:experiment_overview}
    \end{table}
    \subsection{Experiment 1: Linear Baseline and Grover Behavior}
    As a baseline, we consider the linear problem
    \begin{equation}
        u''(x) + 1 = 0.
    \end{equation}
    Two ansatz configurations are studied:
    \begin{itemize}
        \item One-dimensional parameter space defined
        by \(u_1\) \item Two-dimensional parameter space defined by \(u_2\)
    \end{itemize}
    We perform an evaluation run over the Grover iteration count \(k\)for different oracle tolerances
    \(\varepsilon\).

    \noindent These experiments establish that the method successfully identifies low-residual regions
    and that Grover iterations lead to a concentration of probability mass near optimal parameter
    values.
    \paragraph{Linear baseline problem (one-dimensional case)}
    As a reference, we consider the
    linear problem
    \begin{equation}
        u''(x) + 1 = 0,
    \end{equation}
    for which the ansatz \(u_1(x) = w_1\, x(1-x)\) recovers the exact solution at \(w_1 = \tfrac{1}{2}\).

    In the one-dimensional setting, the parameter space is discretized uniformly and contains a single
    parameter \(w_1\). Fora fixed oracle tolerance \(\varepsilon\), the oracle assigns an effective score to
    each parameter value according to the fraction of collocation points at which the residual falls
    below the prescribed threshold.

    Starting from the uniform superposition, amplitude amplificationis applied using \(k\) Grover
    iterations. As predicted by the standard analysis of Grover's algorithm \cite{Grover:1996,
    Nielsen:2010}, the probability mass concentrates on the marked states as \(k\) increases. In the case
    of a unique valid solution, the success probability exhibits the characteristic sinusoidal
    amplification behavior, reaching a maximum at an optimal iteration count as a combination of
    \eqref{def_optimal_k} and \eqref{def_superposition_theta}
    \begin{equation}
        k_{\mathrm{cont}}^\ast =
        \frac{\pi}{4 \arcsin\!\left(\sqrt{K_W/N_W}\right)} -\frac{1}{2}, \qquad k^\ast = \left\lfloor
        k_{\mathrm{cont}}^\ast \right\rfloor .
    \end{equation}
    where \(N_W\) denotes the size of the search
    space. This behavior is clearly reflected in Fig.~\ref{fig:r-test-1dim}, where for decreasing oracle
    tolerance \(\varepsilon\) the set of valid solutions shrinks, eventually approaching the
    single-solution regime. In the case of a unique valid solution (\(K_W=1\)), the success probability
    exhibits the characteristic sinusoidal amplificationbehavior.

    For larger values of \(\varepsilon\), multiple parameter values satisfy the oracle condition, leading
    to a broader set of markedstates and a correspondingly different amplification pattern, inagreement
    with the generalized Grover setting. In the case of \(n_W=6\), the parameter space has cardinality
    \[
        N_W = 2^{n_W} = 64.
    \]
    \begin{figure}
        \centering \includegraphics[width=\textwidth]{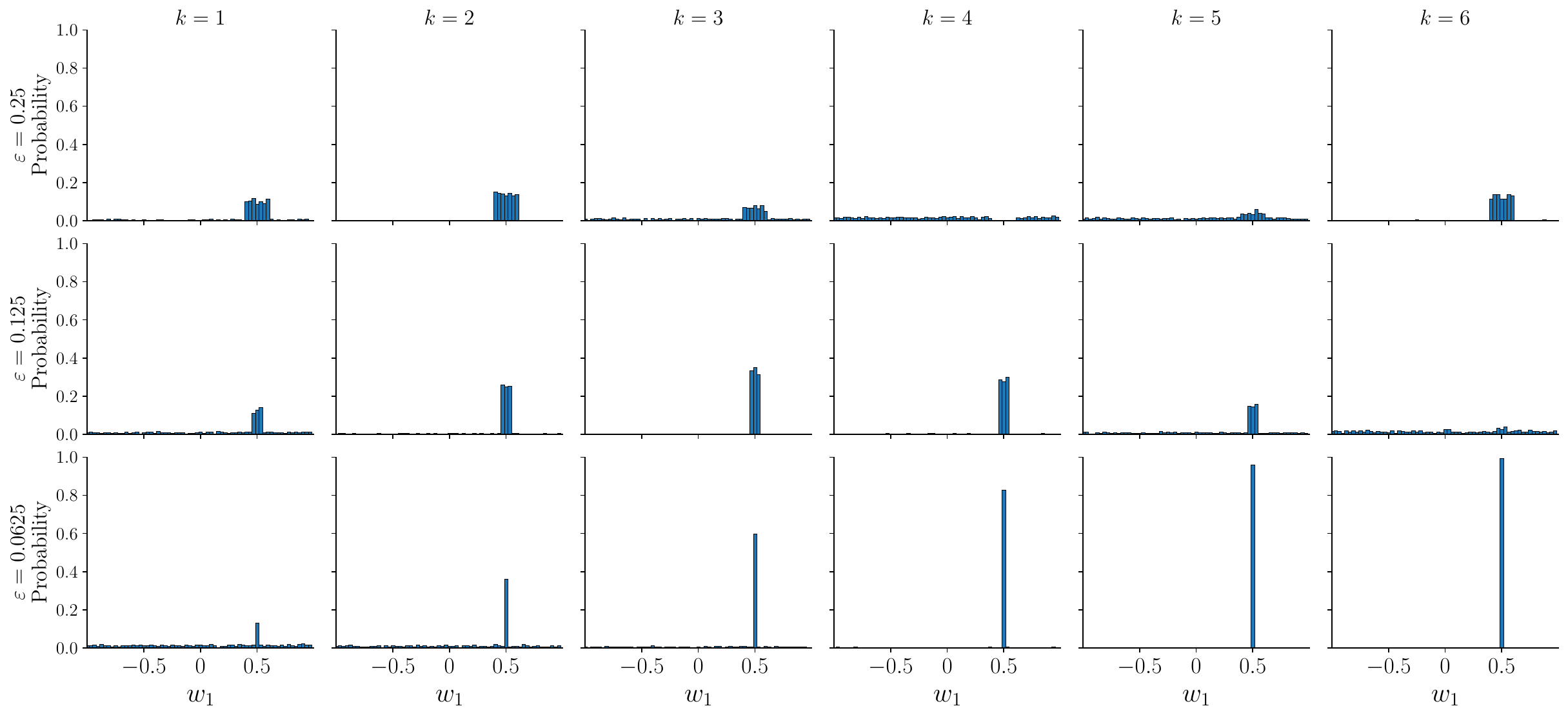}
        \caption{Probability distributions over the one-dimensional parameter space after \(k=1,\dots,6\)
        Grover iterations for different oracle tolerances \(\varepsilon\). As \(\varepsilon\) decreases,
        thenumber of marked states is reduced, leading to a stronger concentration of probability on a
        single parameter value, consistent with the predictions of amplitude amplification.}
        \label{fig:r-test-1dim}
    \end{figure}
    In Fig.~\ref{fig:r-test-1dim}, the number of marked solutions depends on the chosen oracle tolerance
    \(\varepsilon\). For \(\varepsilon = 2^{-1}, 2^{-2}, 2^{-3}\), we observe \(K_W=7,3,1\) marked candidates,
    respectively. Accordingly, the continuous theoretical optima for the number of Grover iterations
    (see Corollary~\ref{cor:standard_grover_limit}) are
    \[
        k_{\mathrm{cont}}^\ast \approx 1.83,\quad
        3.10,\quad 5.77,
    \]
    which correspond to the integer iteration counts
    \[
        k^\ast = 2,\quad 3,\quad 6,
    \]
    after rounding to the nearest integer, in agreement with the theoretical predictions. This
    dependence is reflected in the figure, where the probability mass concentrates more rapidly as the
    number of marked states increases.
    \paragraph{Linear baseline problem (two-dimensional case):}
    \begin{figure}
        \centering \includegraphics[width=\textwidth]{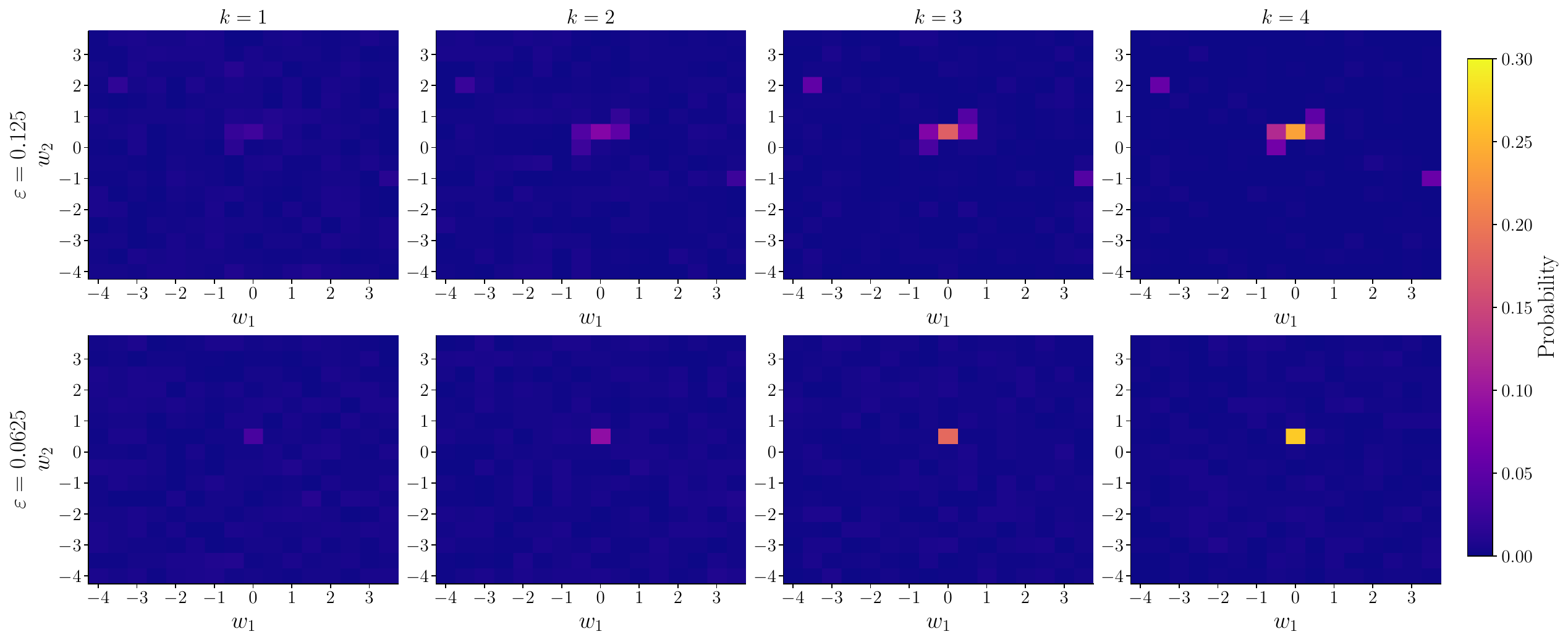}
        \caption{Probability distributions over the two-dimensional parameter space \((w_1,w_2)\) after
        \(k=1,\dots,4\) Grover iterations for different oracle tolerances \(\varepsilon\). In contrast to the
        one-dimensional case, the amplified states form localized regions in parameter space. As
        \(\varepsilon\) decreases, these regions shrink and the probability mass concentrates around the true
        solution, illustrating the extension of amplitude amplification to higher-dimensional parameter
        spaces.} \label{fig:r-test-2dim}
    \end{figure}
    We extend the analysis to a two-dimensional parameter space using the ansatz
    \begin{equation}
        u_2(x)
        = w_2\, x(1-x) + w_1\, x^2 (1-x)^2.
    \end{equation}
    In this setting, the parameter space is
    discretized uniformly with \(n_W = 8\) qubits, resulting in a search space of size $N_W = 2^{n_W} =
    256$.

    As in the one-dimensional case, the oracle marks all parameter pairs \((w_1,w_2)\) whose residual
    falls below a prescribed tolerance \(\varepsilon\). For \(N=256\) and a single marked solution
    (\(K_W=1\)), the success probability after \(k\) Grover iterations is given by
    \begin{equation}
        P(k) =
        \sin^2\!\big((2k+1)\theta\big), \qquad \theta = \arcsin\!\left(\frac{1}{16}\right).
    \end{equation}
    For \(k=4\), this yields a probability of approximately \(0.29\), which is shown by
    Fig.~\ref{fig:r-test-2dim}. In contrast to the one-dimensional case, the solution structure now
    appears as localized regions in a two-dimensional landscape, illustrating that the method naturally
    extends to higher-dimensional parameter spaces.
    \subsection{Experiment 2: Influence of Spatial Resolution}
    In this experiment, we isolate the residual-threshold oracle from the subsequent
    amplitude-amplification dynamics and study the influence of the spatial resolution. Unlike in
    Experiment~1, the oracle is evaluated directly for all \(w\in W_m\), without applying amplitude
    amplification.

    We consider the linear boundary value problem
    \begin{equation}
        \label{eq:experiment2_problem}
        u''(x)+x=0, \qquad x\in[0,1], \qquad u(0)=u(1)=0,
    \end{equation}
    corresponding to
    Eq.~\eqref{eq:test_linear}, with exact solution\(u^\ast(x)=\tfrac16 x(1-x^2)\). We compare the
    expressive ansatz \(u_3(x)=w_2x(1-x^2)\), which is exact for \(w_2^\ast=\tfrac16\), with the
    restricted ansatz \(u_1(x)=w_2x(1-x)\). For \(u_1\), \(r_C(x,w_2)=x-2w_2\), and minimizing the mean
    absolute residualover \([0,1]\) gives \(w_2^\ast=\tfrac14\).

    We fix the parameter-space discretization to \(n_W=4\), the target-register precision to \(p_Z=5\),
    and choose \(n_Z\) according to the heuristic estimate in Eq.~\eqref{nZguess}. Only the spatial
    resolution \(n_X\) is varied.
    \subsubsection{Residual Measures}
    For a pointwise residual \(r(x_i,w)\), we define the aggregated collocation residual by
    \begin{equation}
        \label{eq:aggregated_collocation_residual} \bar r(w) = \frac{1}{N_X}
        \sum_{i=0}^{N_X-1} |r(x_i,w)|.
    \end{equation}
    We distinguish the analytical reference residual
    \(\bar r_C\), obtained using the exact derivative of the ansatz; the finite-difference residual
    \(\bar r_{FD}\), obtained using \(D_h^2u\); and the quantum residual \(\bar r_Q\), evaluated by the
    reversible fixed-point arithmeticcircuit.

    Since all three residuals are evaluated on the same spatial and parameter grids, the difference
    between \(\bar r_C\) and \(\bar r_{FD}\) captures effects of the discrete residual evaluation,
    whereas the difference between \(\bar r_{FD}\) and \(\bar r_Q\) isolates the additional fixed-point
    quantization error. We further use
    \begin{equation}
        r_C^{\min}(w) = \min_i |r_C(x_i,w)|, \qquad
        r_C^{\max}(w) = \max_i |r_C(x_i,w)|
    \end{equation}
    to characterize the spatial variation of the
    analytical residual.

    For both ansatz functions considered below, the centered finite-difference approximation of the
    second derivative is exact at interior grid points. The dependence on \(n_X\) therefore results from
    spatial sampling, aggregation over collocation points, boundary treatment, and fixed-point
    quantization rather than finite-difference truncation error.
    \subsubsection{Ansatz Function \(u_3\)}
    For \( u_3(x)=w_2x(1-x^2), \) the analytical residual is \( r_C(x,w_2)=(1-6w_2)x. \) Hence, the
    exact solution is contained in the ansatz space and is attained at \(w_2^\ast=\tfrac16\). Since this
    value is not represented exactly on the chosen parameter grid, the nearest admissible value is
    \(w_2=\tfrac18\).

    Figure~\ref{fig:nX-test} shows that all three residual evaluations attain their minimum near this
    discretized optimum. The agreement between \(\bar r_{FD}\) and \(\bar r_C\) confirms the exactness
    of the centered finite difference for this cubic ansatz, while deviations of \(\bar r_Q\) from
    \(\bar r_{FD}\) quantify the remaining fixed-point effects.
    \begin{figure}[!htbp]
        \centering \includegraphics[width=\textwidth]
        {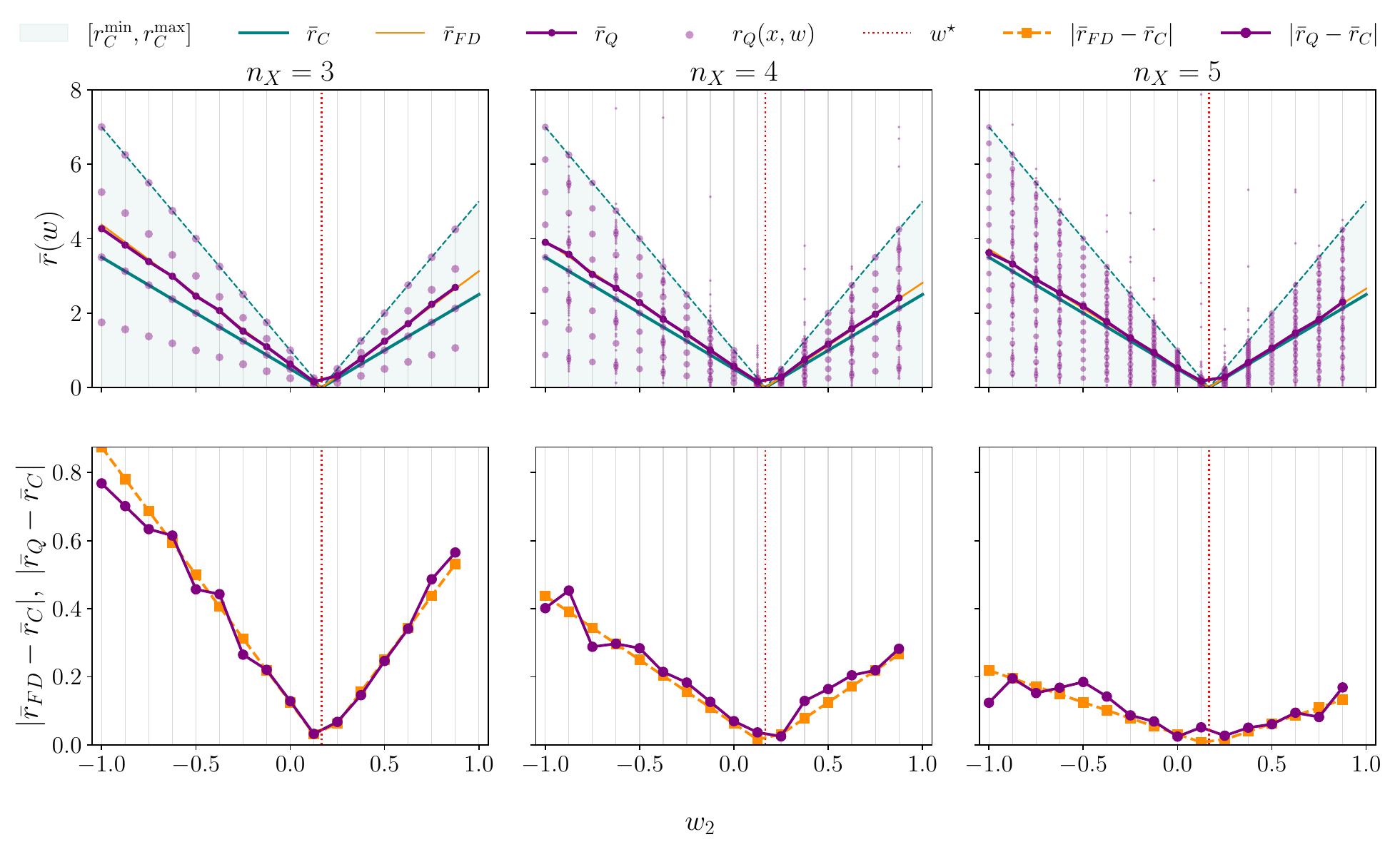} \caption{ Residual landscape as a function of \(w_2\) for
        increasing spatial resolution \(n_X\), using the expressive ansatz \(u_3\). The exact optimum is
        \(w_2^\ast=1/6\), while \(w_2=1/8\) is the nearestadmissible value on the chosen parameter grid.
        Differences between the finite-difference and quantum residuals reflect fixed-point quantization
        effects. } \label{fig:nX-test}
    \end{figure}
    \subsubsection{Ansatz Function \(u_1\)}
    For \( u_1(x)=w_2x(1-x), \) the analytical residual is \( r_C(x,w_2)=x-2w_2. \) The exact solution
    is not contained in this ansatz space. Minimizing the mean absolute residual over \(x\in[0,1]\)
    gives \(w_2^\ast=\tfrac14\), rather than \(\tfrac16\), because \(2w_2\) must equal the median of the
    interval. The residual minimum therefore remains nonzero even at the optimal parameter value.

    As shown in Fig.~\ref{fig:nX-test-v2}, the quantum residual closely follows the finite-difference
    residual, whereas the nonvanishingminimum is caused by the limited expressivity of the ansatz rather
    than by the reversible implementation.
    \begin{figure}[!htbp]
        \centering \includegraphics[width=\textwidth]
        {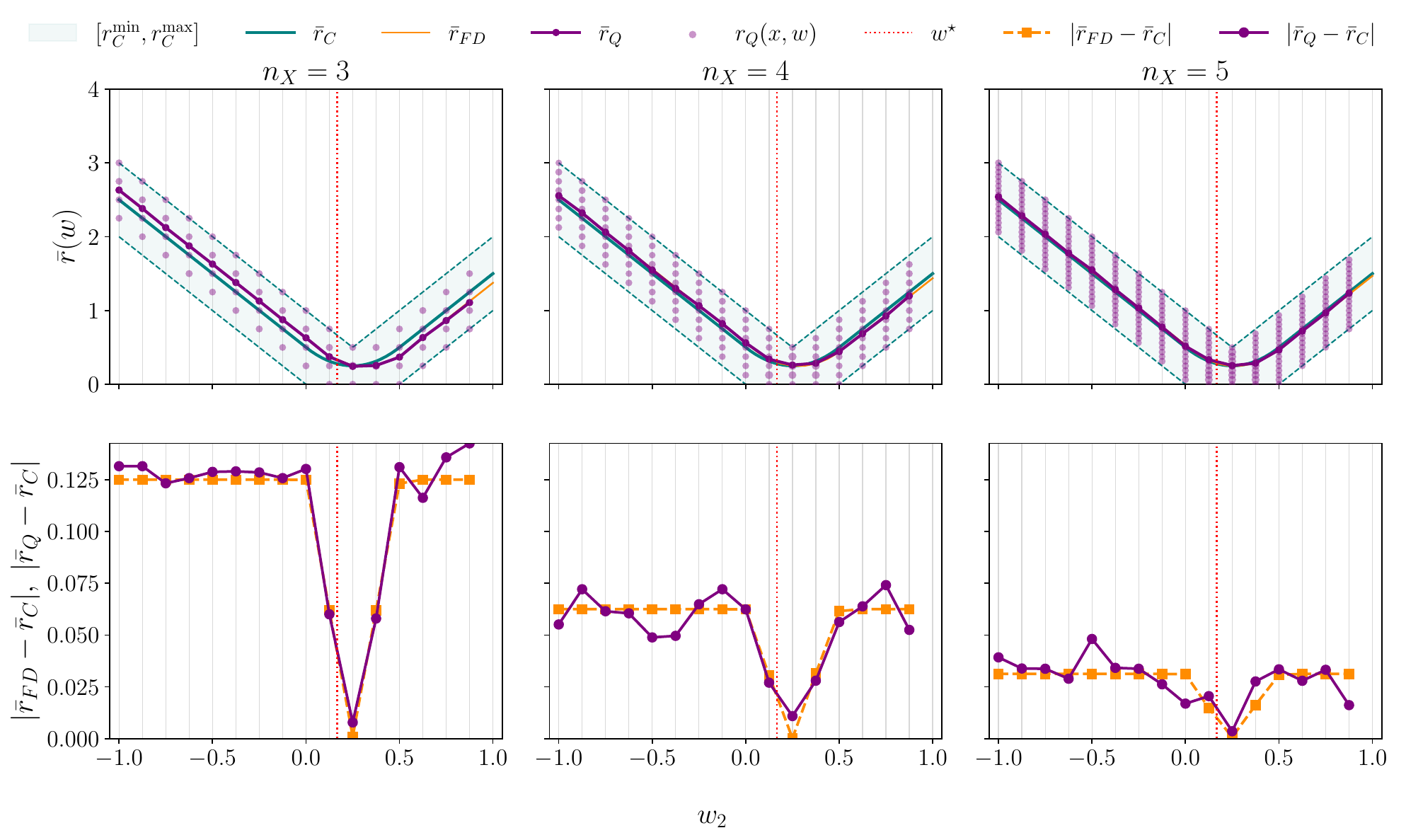} \caption{ Residual landscape as a function of \(w_2\) for
        increasing spatial resolution \(n_X\), using the restricted ansatz \(u_1\). Theresidual-minimizing
        parameter approaches \(w_2=1/4\), while theminimum remains nonzero because the exact solution is not
        contained in the ansatz space. The quantum residual closely follows the finite-difference residual.
        } \label{fig:nX-test-v2}
    \end{figure}
    \subsubsection{Comparison of Residual-Evaluation Errors}
    To quantify deviations from the analytical reference, we use
    \begin{equation}
        \label{eq:experiment2_mape} \operatorname{MAPE}(\bar r,\bar r_C) = \frac{100\%}{N_W} \sum_{w\in W_m}
        \frac{|\bar r(w)-\bar r_C(w)|}{|\bar r_C(w)|}.
    \end{equation}
    For the chosen parameter grid, the
    denominator does not vanish for any evaluated parameter value.

    Figure~\ref{fig:nX-mape} compares the MAPE of \(\bar r_{FD}\) and \(\bar r_Q\) for both ansatz
    functions. The finite-difference error reflects the discrete collocation and boundary treatment,
    while the additional quantum error is caused by fixed-point quantization.
    \begin{figure}[!htbp]
        \centering \begin{minipage}{0.48\textwidth} \centering
        \includegraphics[width=\textwidth] {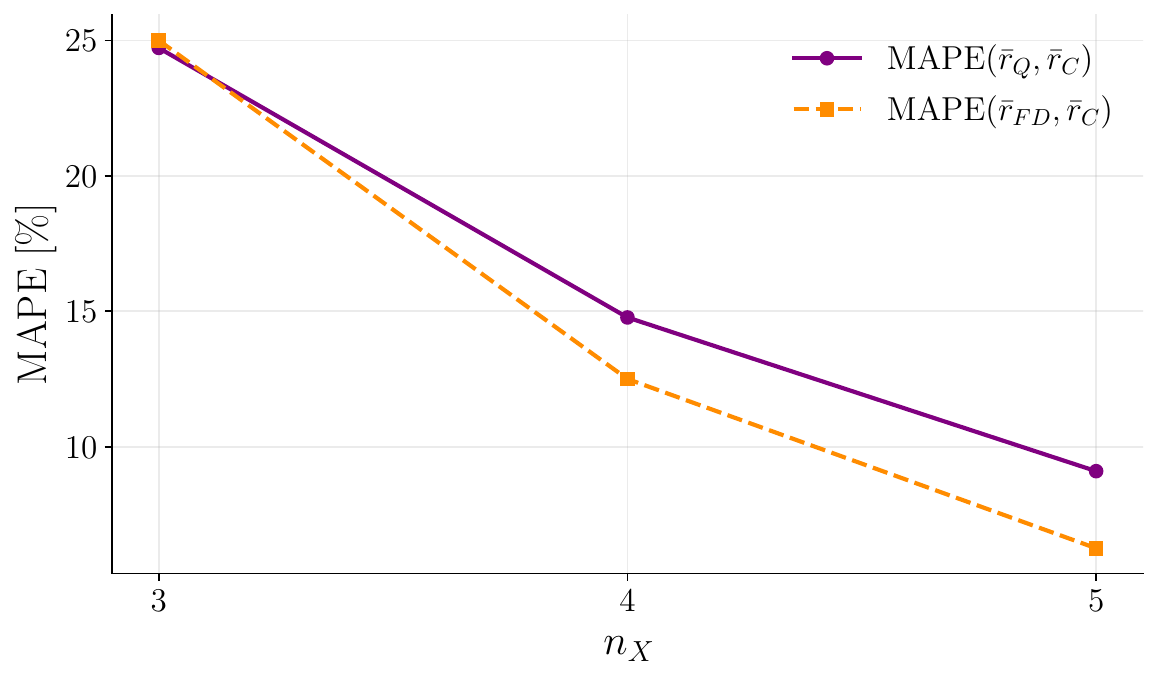} \end{minipage} \hfill
        \begin{minipage}{0.48\textwidth} \centering \includegraphics[width=\textwidth]
            {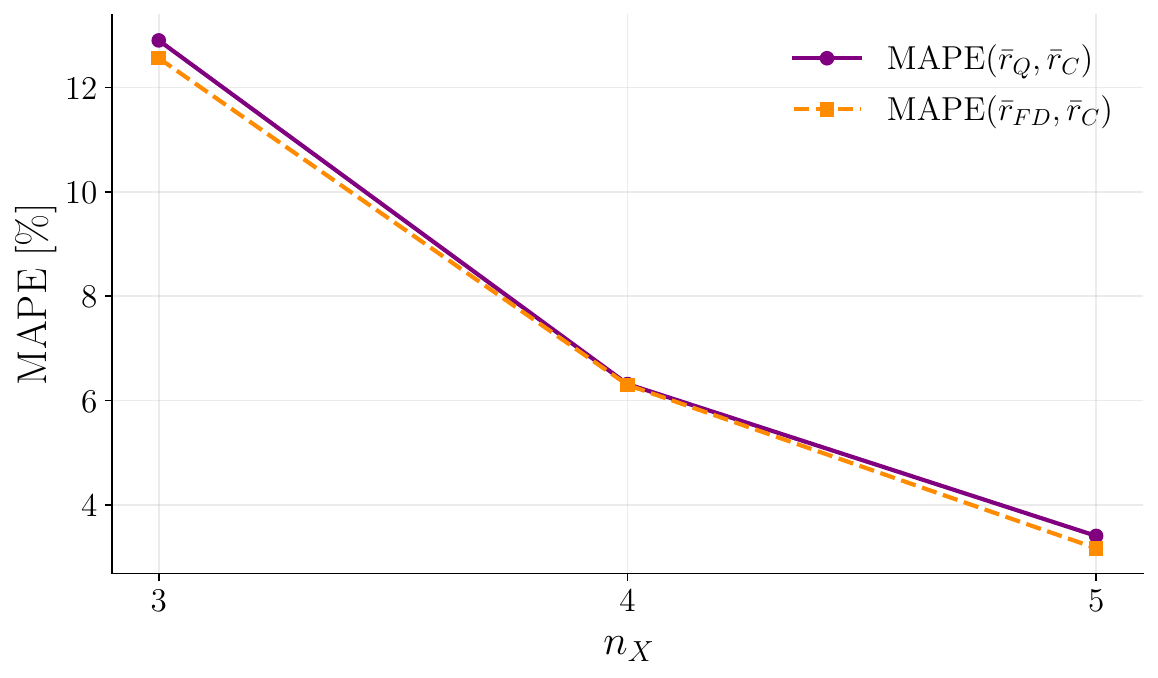} \end{minipage} \caption{ Mean absolute percentage error of the
            finite-difference and quantum residuals relative to the analytical reference for increasingspatial
            resolution \(n_X\). The left and right panels correspond to the expressive ansatz \(u_3\) and the
            restricted ansatz \(u_1\),respectively. } \label{fig:nX-mape}
        \end{figure}
        The results separate three effects: the parameter-grid error for\(u_3\), the intrinsic ansatz error
        for \(u_1\), and the fixed-point quantization error of the reversible residual evaluation. We next
        study how the resulting oracle phases affect the amplitude-amplification dynamics.
        \subsubsection{Empirical Verification of the Amplification Dynamics}
        According to Theorem~\ref{theorem:coherent_amplitude_amplification}, each spatial state \(x_i\)
        evolves with the effective Grover angle
        \[
            \theta_i = \arcsin\!\left( \sqrt{\frac{|\mathcal
            M_i|}{N_W}} \right), \qquad \mathcal M_i = \{j\in\mathcal I_W\mid \chi(i,j)=1\},
        \]
        where \(\mathcal
        M_i\) denotes the set of marked parameter indices at the collocation point \(x_i\). Larger marked
        sets therefore leadto faster amplification, whereas smaller sets produce slower but more selective
        dynamics.

        We consider \(n_X=3\), \(n_W=5\), and \(\varepsilon\in\{0.125,0.25,0.5\}\). For each \(w_j\in W_m\),
        the weighted-oracle response \(q(w_j)\) is the fraction of collocation points satisfying the
        residual threshold. Hence,
        \[
            \frac{1}{N_W} \sum_{j=0}^{N_W-1}q(w_j)
        \]
        gives the average oracle
        response over the parameter register.

        For \(\varepsilon=0.5,0.25,0.125\), the marked-set sizes are \(|\mathcal M_i|=7,3,1\), respectively,
        for every spatial sector. Thus, all sectors have identical amplification angles, and the dynamics
        reduce to the standard Grover regime of Corollary~\ref{cor:standard_grover_limit}. The corresponding
        continuous optimal iteration counts,
        \[
            k_{\mathrm{cont}}^\ast = \frac{\pi}{4\theta_i} -\frac12,
        \]
        are approximately \(1.1\), \(2.0\), and \(3.9\).

        Figure~\ref{fig:weighted-grover-example} shows the residual functions, marked regions, and measured
        parameter distributions after \(k=1,2,3\) amplification iterations. Decreasing
        \(\varepsilon\)narrows the marked region and reduces the amplification angle. Larger tolerances
        therefore lead to faster amplification and earlier overrotation, whereas smaller tolerances yield
        slower but more selective concentration near \(w^\ast=\tfrac16\). The observed dynamics agree with
        Theorem~\ref{theorem:coherent_amplitude_amplification}.
        \begin{figure}[!htbp]
            \centering \includegraphics[width=1\textwidth]
            {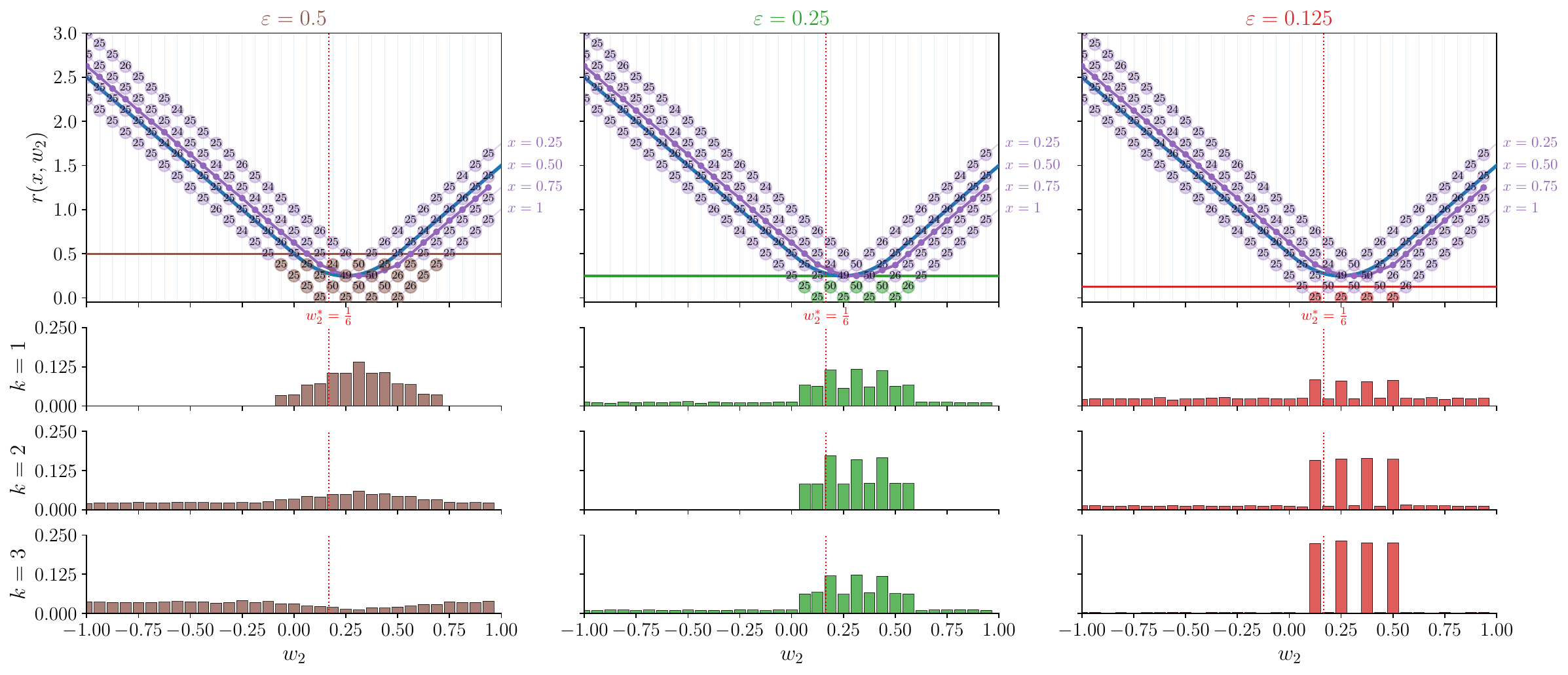} \caption{ Empirical spatially conditioned amplitude
            amplification for different oracle tolerances \(\varepsilon\). The purple curves show the residual
            functions \(r(x_i,w_2)\), and the labels indicate the empirical sampling probabilities of the
            corresponding parameter values. The top row shows the marked regions with \(|\mathcal M_i|=7,3,1\)
            from left to right, while the lower rows show the measured parameter distributions after \(k=1,2,3\)
            amplification iterations. } \label{fig:weighted-grover-example}
        \end{figure}
        \subsection{Experiment 3: Nonlinear Residuals and Finite-Precision Effects}
        We consider the manufactured nonlinear boundary value problem
        \begin{equation}
            u''(x) + u(x) +
            \alpha u(x)^2 + 2\lambda - \lambda x(1-x) - \alpha \lambda^2 x^2(1-x)^2 = 0,
        \end{equation}
        for \(
        x\in[0,1] \) with homogeneous boundary conditions \( u(0)=u(1)=0 \) and nonlinearity parameters \(
        \alpha\in\{8,16,32\}. \) By construction, \( u^\ast(x) = \lambda x(1-x) \) is an exact solution.
        Throughout this experiment, we set \( \lambda=\frac13. \)

        In contrast to the preceding experiments, the residual now contains the nonlinear term \( u(x)^2 \).
        The goal of this experiment is therefore twofold:
        \begin{enumerate}
            \item to analyze how increasing
            nonlinearity modifies the residual landscape and the resulting amplification behavior,

            \item to quantify the influence of finite-precision arithmetic on the oracle evaluation.
        \end{enumerate}
        \subsubsection{Nonlinearity-Induced Residual Structure}
        We first fix the value-register precision to \( p_Z=5 \) and vary the nonlinearity parameter \(
        \alpha \). For weak nonlinearity, the residual landscape is sharply localized around the physical
        solution \(w^\ast=\frac13\). As \(\alpha\) increases, additional low-residual regions and local
        minima emerge.

        Although the problem is constructed so that \(u^\ast(x)=\lambda x(1-x)\) remains an exact solution,
        the nonlinear residual creates competing parameter regions with comparable residual values.
        Consequently, the oracle assigns significant response toseveral regions rather than marking a single
        isolated one.

        This structure is reflected in the amplified sampling distributions in the lower panels of
        Fig.~\ref{fig:quantization-effects}. Larger tolerances \(\varepsilon\) produce broader amplified
        regions, whereas smaller tolerances increase selectivity and suppress parts of the secondary minima.
        Even for strong nonlinearity, \(\alpha=32\), the amplification process reliably identifies the
        dominant low-residual regions.

        This behavior agrees with the weighted-oracle framework of
        Theorem~\ref{theorem:coherent_amplitude_amplification}, in whichseveral parameter regions with large
        local oracle response \(q(w_j)\) are amplified simultaneously.
        \subsubsection{Finite-Precision
        Quantization Effects}
        We now fix the nonlinearity parameter to \( \alpha=8 \) and investigate the influence of
        finite-precision arithmetic by varying the value-register precision
        \[
            p_Z\in\{3,4,5\}.
        \]
        As an empirical measure of quantization artifacts, we introduce the out-of-band probability
        \begin{equation}
            P_{\mathrm{OOB}}(w) = 100 \sum_{o\in\mathcal O(w)} p(o\mid w)\, \mathbf1\!\left[
            |r_Q(o,w)| \notin \left[ r_C^{\min}(w)-\delta_{\mathrm{band}}, r_C^{\max}(w)+\delta_{\mathrm{band}}
            \right] \right],
        \end{equation}
        which measures the fraction of quantum outcomes whose residual
        values lie outside the classical reference interval up to a small tolerance \(
        \delta_{\mathrm{band}} \). This tolerance prevents minor numerical deviations near the interval
        boundaries from being counted as meaningful discrepancies.
        \begin{figure}[!htbp]
            \centering
            \includegraphics[width=\textwidth]{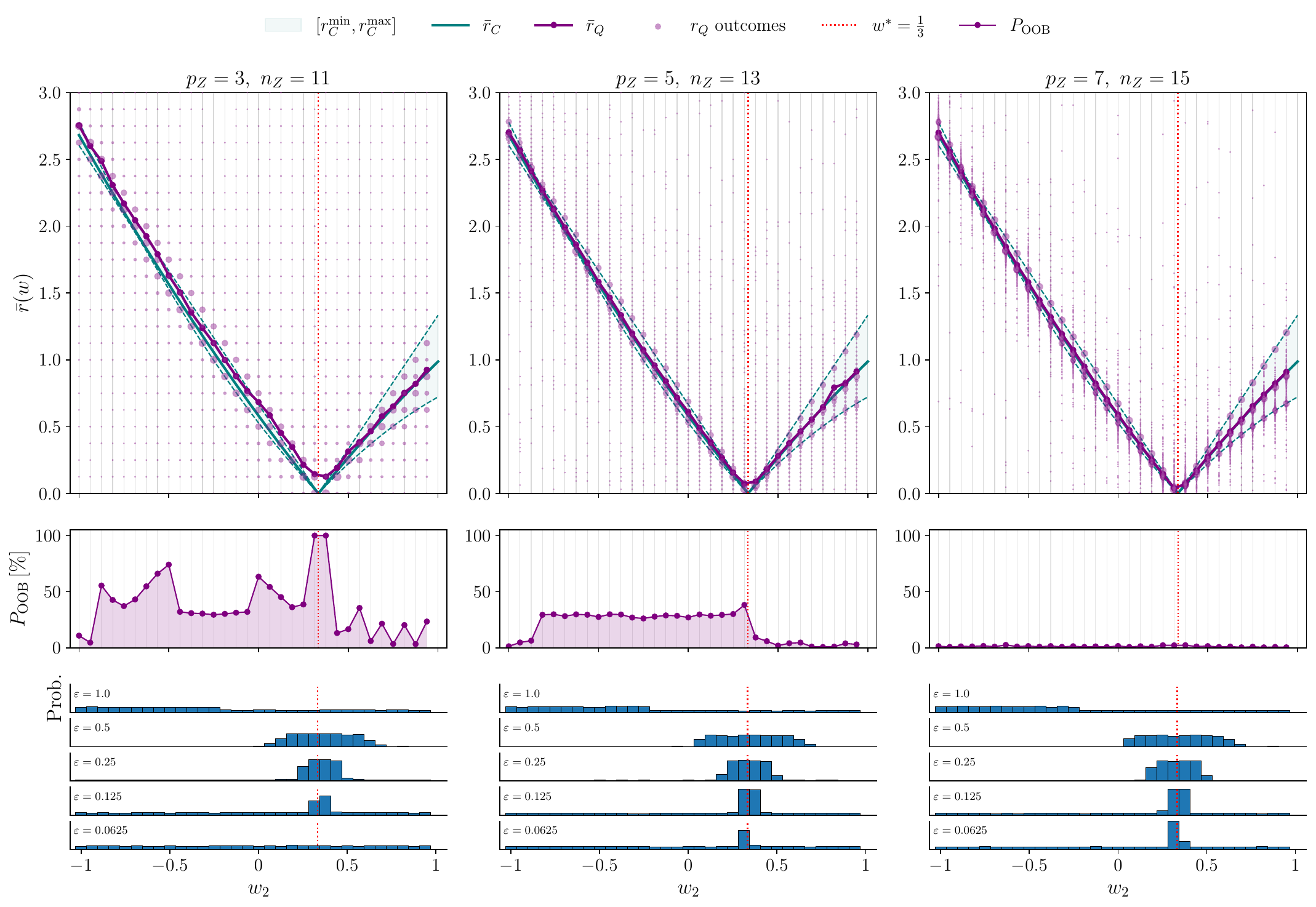}
            \caption{ Influence of finite-precision arithmetic for different value-register precisions \(p_Z\).
            Top: classical and quantum residual landscapes together with sampled measurement outcomes. Middle:
            out-of-band probability \(P_{\mathrm{OOB}}\). Bottom: amplified parameter distributions for
            different oracle tolerances \(\varepsilon\). Increasing \(p_Z\) reduces quantization artifacts. }
            \label{fig:quantization-effects}
        \end{figure}
        \begin{figure}[!htbp]
            suppresses quantization-induced deviations in the 
        \end{figure}
        The results are shown in Fig.~\ref{fig:quantization-effects}. Atlow precision, the quantum residual
        evaluation visibly deviates fromthe classical finite-difference reference because of rounding and
        limited residual resolution, as reflected by elevated values of \(P_{\mathrm{OOB}}\). Increasing the
        value-register precision systematically reduces these deviations and drives the quantum residual
        toward the classical discretized residual. Meanwhile, the dominant amplified parameter region
        remains stable, indicating that the amplification process is robust against moderate quantization
        errors. At very low precision, however, the residual grid may be too coarse for the prescribed
        oracle tolerance. For \(p_Z=3\) and \(\varepsilon=2^{-4}\), the residual spacing is \(2^{-3}\), so
        the threshold \(2^{-4}\) is not representable. Consequently, no parameter value satisfies the
        threshold condition and the oracle marks no admissible candidate. This failure is therefore caused
        not by quantization noise itself, but by insufficient fixed-point resolution. The oracle tolerance
        must remain compatible with the representable residual precision.
        \subsection{Experiment 4: Resource Scaling of the Oracle Implementation}
        After validating the oracle accuracy and amplification behavior,we analyze the resource requirements
        of the resulting quantum circuits. Since multi-controlled phase operations are not native to current
        gate-based quantum hardware, all circuits were transpiled into the native gate basis prior to
        resource estimation.

        For the scaling analysis, we vary the spatial discretization parameter \(n_X\) while keeping \(p_W =
        3\), \(p_Z = 3\), and the ansatz configuration fixed. For each transpiled circuit, we record the
        native single-qubit gate count, native two-qubit gate count, and total native gate count.

        The observed scaling behavior is consistent with the theoreticaloracle-cost estimate derived in
        Theorem~\ref{thm:oracle_cost},
        \[
            C_{\mathrm{oracle}} = O\!\left( n_Z n_X^{D_f} n_W^{Q} (D_f+Q+1)^2
            + n_Z^2 \right).
        \]
        which predicts polynomial growth with respect to the spatial discretization and parameter-space
        resolution.

        Figure~\ref{fig:nops-analysis} shows the resulting scaling behavior for the benchmark problems
        considered in this work. Dotted lines denote native single-qubit gates, dashed lines native
        two-qubit gates, and solid lines the total native gate count after transpilation. Although the
        absolute gate counts differ between problem classes, all cases exhibit clear polynomial growth as
        \(n_X\) increases.
        \begin{figure}[!htbp]
            \centering
            \includegraphics[width=1.0\textwidth]{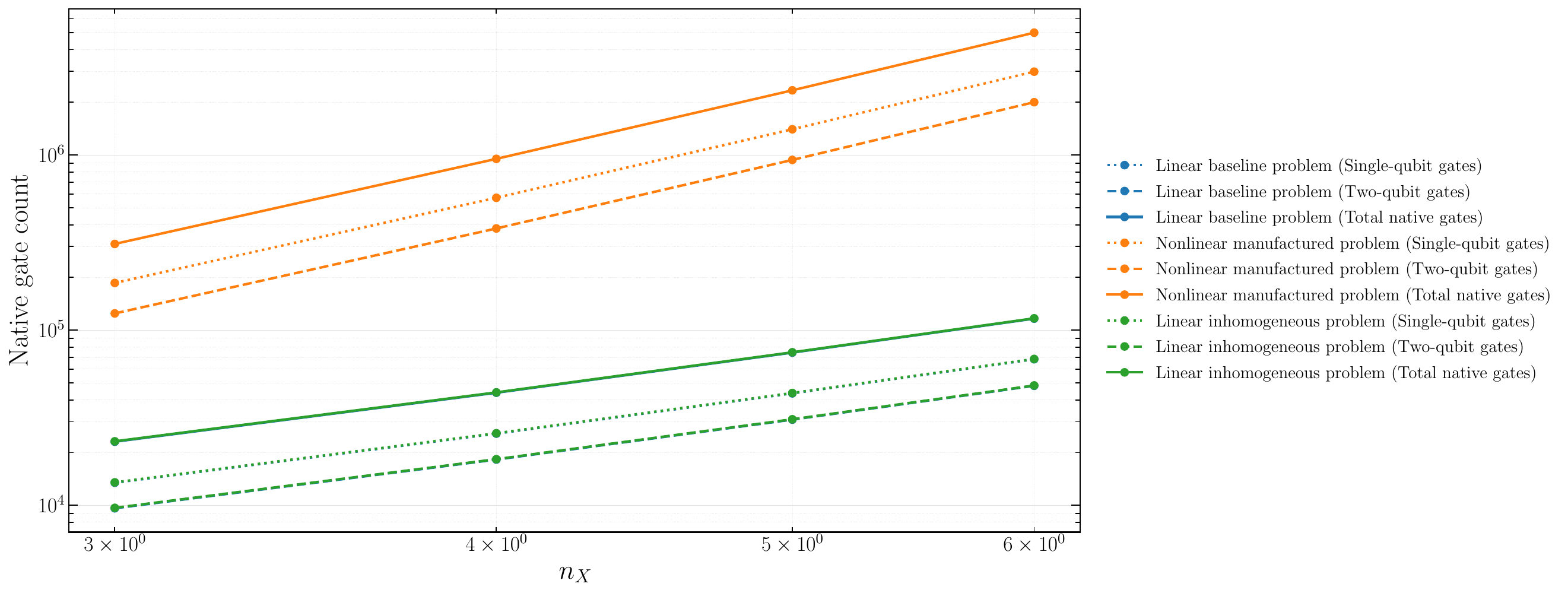}

            \caption{ Scaling behavior of the transpiled oracle implementation cost asa function of the spatial
            discretization parameter \(n_X\) for thedifferent benchmark problems considered in this work. All
            oraclecircuits were transpiled intothe native gate basis before evaluation. Dotted lines denote
            native single-qubit gate counts, dashed lines native two-qubit gate counts, and solid lines the
            total native gate count. } \label{fig:nops-analysis}
        \end{figure}
        \section{Conclusion}
        \label{sec:conclusion} We have introduced a residual-based quantum collocation
        framework for one-dimensional boundary value problems and shown that its amplification dynamics
        differ fundamentally from the standard Grover setting. Rather than being governed by a single global
        marking condition, the search process evolves through a coherentsuperposition of spatially
        conditioned amplitude amplification rotations, with success probabilities determined by the
        distribution of local residual-threshold responses across the collocation grid.

        The numerical experiments support the theoretical predictions and illustrate how discretization,
        ansatz expressivity, oracle tolerance, and finite-precision arithmetic influence both approximation
        quality and amplification behavior. Furthermore, the results show that the reversible residual
        oracle scales polynomially in the logarithmic size of the spatial register, while the parameter
        search benefits from the usual quadratic speedup provided by amplitude amplification. The present
        study focuses on low-dimensional polynomial ansatz spaces chosen to permit compact oracle
        constructions and modest qubit requirements, while the general framework is not restricted to these
        particular choices.Several directions for future research remain open. These include extensions to
        higher-dimensional partial differential equations, the investigation of more expressive ansatz
        spaces, adaptive collocation strategies, and scalability studies of the reversible oracle
        construction. Of particular interest are high-dimensional nonlinear PDEs arising in optimal control
        and stochastic control, such as Hamilton--Jacobi--Bellman equations,where classical
        discretization-based methods are often limited by the curse of dimensionality. Equally important is
        the study of the proposed framework on quantum hardware, where noise, decoherence and finite
        sampling effects may significantly influence the residual-based amplification dynamics. Overall, the
        framework establishes a foundation for residual-based quantum search methods for differential
        equations and highlights a novelconnection between collocation techniques and coherent
        amplitudeamplification.

        \appendix

        \section{Quantum Registers and Discretized Spaces}
        \label{subsec_quantum_registers}
        \paragraph{State register}
        Let \(X\) be an \(n_X\)-qubit register with associated Hilbert space
        \(\mathcal H_X := (\mathbb C^2)^{\otimes n_X}\) with \(N_X =2^{n_X}\). Its computational basis is
        \[
            \{\ket{i}_X\}_{i=0}^{N_X-1}, \qquad \langle i\mid j\rangle_X=\delta_{ij},
        \]
        where, for the binary
        expansion
        \[
            i=\sum_{j=0}^{n_X-1} b_j2^j, \qquad b_j\in\{0,1\},
        \]
        we define
        \[
            \ket{i}_X :=
            \ket{b_{n_X-1}}\otimes\cdots\otimes\ket{b_0}.
        \]
        The register uses a fixed-point representation with \(f_X\) fractional bits. Thus, the same
        bitstring encodes the numerical value
        \[
            x=\sum_{j=0}^{n_X-1} b_j\,2^{j-f_X}, \qquad \ket{x}_X :=
            \ket{b_{n_X-1}}\otimes\cdots\otimes\ket{b_0}.
        \]
        We write \(\ket{i}_X\) when treating computational
        basis states as indexed elements of \(\mathcal H_X\), and \(\ket{x}_X\) when interpreting them as
        encoded numerical values.

        The spatial grid is a subset
        \[
            X\subset [0,1]\cap 2^{-f_X}\mathbb N_0.
        \]
        Not every computational
        basis state need correspond to an element of \(X\). In particular, the register may encode values
        slightly outside the physical domain \([0,1]\), as required for the reversible evaluation of
        finite-difference expressions such as \(x\pm h\).

        A pure state of the register is a normalized vector
        \[
            \ket{\psi}_X =
            \sum_{i=0}^{N_X-1}\alpha_i\ket{i}_X, \qquad \alpha_i\in\mathbb C, \qquad
            \sum_{i=0}^{N_X-1}|\alpha_i|^2=1.
        \]
        \paragraph{Parameter register}
        Let
        \[
            W_m \subset (2^{-f_w}\mathbb Z)^m
        \]
        denote the set of all parameter vectors representable
        with the chosen fixed-point encoding for an \(m\)-parameter ansatz. A parameter vector \(
        w=(w_1,\dots,w_m)\in W_m \) consists of \(m\) components, each encoded as a signed fixed-point
        number with \(I_w\) integer bits and \(f_w\) fractional bits. Including the sign bit, each component
        requires \( n_w = 1 + I_w + f_w \) qubits.

        If \( b^\ell_{n_w-1}\dots b^\ell_0 \) denotes the bitstring encoding the component \(w_\ell\), then
        \[
            w_\ell = -\,b^\ell_{n_w-1}2^{\,n_w-1-f_w} + \sum_{s=0}^{n_w-2} b^\ell_s\,2^{\,s-f_w}, \qquad
            b^\ell_s\in\{0,1\}.
        \]
        which corresponds to the standard two's-complement fixed-point
        representation.

        The parameter register therefore consists of \(m\) blocks of \(n_w\) qubits and contains \( n_W = m
        n_w \) qubits in total.

        Its Hilbert space is \( \mathcal H_W := (\mathbb C^2)^{\otimes n_W} \) with computational basis
        \[
            \{\ket{j}_W\}_{j=0}^{N_W-1}, \qquad N_W = 2^{n_W}.
        \]
        Each basis state \(\ket{j}_W\) encodes a parameter vector
        \[
            w_j = (w_{j,1},\dots,w_{j,m}) \in W_m,
        \]
        where each component \(w_{j,\ell}\) is obtained by decoding the corresponding block of \(n_w\)
        qubits. \medskip
        \paragraph{Value register}
        Let the value register consist of \(n_Z\) qubits with \(f_Z\) fractional bits. For
        \(k\in\{0,\ldots,N_Z-1\}\), let
        \[
            k=\sum_{j=0}^{n_Z-1} b_j2^j, \qquad b_j\in\{0,1\},
        \]
        be the binary
        representation associated with the computational basis state \(\ket{k}_Z\). The corresponding signed
        fixed-point value is
        \[
            z_k = - b_{n_Z-1} 2^{\,n_Z-1-f_Z} + \sum_{j=0}^{n_Z-2} b_j\,2^{\,j-f_Z}.
        \]
        The Hilbert space of the value register is \( \mathcal H_Z := (\mathbb C^2)^{\otimes n_Z} \) with
        computational basis
        \[
            \{\ket{k}_Z\}_{k=0}^{N_Z-1}, \qquad N_Z = 2^{n_Z}.
        \]
        We denote by
        \begin{align}
            \label{def_value_register} Z=\{z_k\}_{k=0}^{N_Z-1}\subset 2^{-f_Z}\mathbb
            Z
        \end{align}
        the set of all values representable by this fixed-point encoding. \medskip
        \section{QFT-based arithmetic}
        \label{subsec:qft_polynomial}
        We use standard QFT-based reversible arithmetic primitives for modular addition and multiplication;
        see, e.g., \cite{VedralBarencoEkert:1996,Draper:2000,RuizPerez:2017}. For the residual oracle, we
        need only the resulting scaling for polynomial terms. We therefore record the following estimate for
        monomials of the form \(x^{k_x}w^{k_y}\).

        Let \(X\), \(W\), and \(Z\) be computational registers of sizes \(n_X\), \(n_W\), and \(n_Z\),
        respectively, and let \(N_Z=2^{n_Z}\). We consider the reversible modular map
        \begin{equation}
            \label{eq:poly_map} \ket{x}_X\ket{w}_W\ket{z}_Z \longmapsto \ket{x}_X\ket{w}_W \ket{z\oplus
            x^{k_x}w^{k_y}}_Z ,
        \end{equation}
        where the addition is modulo \(N_Z\). Throughout this estimate,
        arithmetic is understood at the level of the encoded integers; fixed-point scaling factors only
        change the rotation angles and not the asymptotic gate count.

        Applying the QFT to the target register \(Z\) turns the modular addition in \eqref{eq:poly_map} into
        multiplication by the phase
        \[
            \exp\!\left(2\pi i\,\frac{x^{k_x}w^{k_y}q}{N_Z}\right),
        \]
        where
        \(q\) denotes the Fourier-basis index of the \(Z\)-register. Expanding \(x\), \(w\), and \(q\) in
        binary expresses this phase as a product of multi-controlled phase rotations. Each contribution is
        controlled on at most \(k_x\) bits from the \(X\)-register, \(k_y\) bits from the \(W\)-register,
        and one Fourier bit of the \(Z\)-register. Repeated bit indices only merge or reduce controls and
        therefore do not affect the upper bound below.
        \begin{lemma}[Gate complexity of QFT-based polynomial arithmetic]
            \label{lemma:qft_polynomial_runtime} Let \(U_{\mathrm{poly}}\) denote the QFT-based reversible
            implementation of a monomial contribution
            \[
                c\,x^{k_x}w^{k_y}
            \]
            to the modular addition in
            \eqref{eq:poly_map}, where \(x\) and \(w\) are scalar variables encoded in \(n_X\) and \(n_W\)
            qubits, respectively. Counting arbitrary one-qubit gates and CNOT gates as elementary, as in Barenco
            et al.~\cite{Barenco:1995}, one has
            \[
                C(U_{\mathrm{poly}}) = O\!\left( n_Z
                n_X^{k_x}n_W^{k_y}(k_x+k_y+1)^2+n_Z^2 \right).
            \]
        \end{lemma}
        \begin{proof}
            Write
            \[
                x=\sum_{i=0}^{n_X-1}2^i x_i, \qquad w=\sum_{j=0}^{n_W-1}2^j w_j, \qquad
                x_i,w_j\in\{0,1\}.
            \]
            Then
            \[
                x^{k_x} = \sum_{i_1=0}^{n_X-1}\cdots \sum_{i_{k_x}=0}^{n_X-1}
                2^{i_1+\cdots+i_{k_x}} x_{i_1}\cdots x_{i_{k_x}},
            \]
            and analogously for \(w^{k_y}\). Hence,
            \(x^{k_x}w^{k_y}\) is a sum of at most \(n_X^{k_x}n_W^{k_y}\) bit monomials. Since repeated indices
            maybe removed using \(x_i^r=x_i\) and \(w_j^r=w_j\), each monomial depends on at most \(k_x+k_y\)
            distinct input bits.

            In the Fourier basis, each bit monomial induces at most one phase rotation on each of the \(n_Z\)
            target qubits. Thus, at most
            \[
                n_Z n_X^{k_x}n_W^{k_y}
            \]
            multi-controlled one-qubit phase gates are
            required. Each has atmost \(c=k_x+k_y\) controls. By Corollary~7.6 of Barenco et
            al.~\cite{Barenco:1995}, such a gate can be decomposed into
            \[
                O\!\left((c+1)^2\right) =
                O\!\left((k_x+k_y+1)^2\right)
            \]
            elementary one-qubit and CNOT gates. The phase accumulation
            therefore has complexity
            \[
                O\!\left( n_Z n_X^{k_x}n_W^{k_y}(k_x+k_y+1)^2 \right).
            \]
            Adding the
            \(O(n_Z^2)\) cost of the QFT and inverse QFT gives the stated bound.
        \end{proof}
        For signed fixed-point encodings, the sign bits only change the signs of the corresponding phase
        angles. The control structure, and hencethe asymptotic gate count in
        Lemma~\ref{lemma:qft_polynomial_runtime}, remains unchanged.

        \section{Reversible Implementations}
        \label{section:reversible_implementations}

        This appendix gives explicit reversible constructions of the residual, absolute-value, and
        threshold-subtraction operators \(U_r\), \(U_a\), and \(U_s\) entering
        \[
            U_c = (I_X\otimes
            I_W\otimes U_s\otimes I_A) (I_X\otimes I_W\otimes U_a) (U_r\otimes I_A).
        \]
        The register encodings
        are specified in Appendix~\ref{subsec_quantum_registers}.
        \subsection*{Reversible Absolute-Value Operator}
        \label{implementation_absolute_operator}
        The \(n_Z\)-qubit value register uses two's-complement fixed-point encoding with \(f_Z\) fractional
        bits. Its representable set \(Z\) is defined in \eqref{def_value_register}, with extremal values
        \[
            z_{\min}=-2^{\,n_Z-1-f_Z}, \qquad z_{\max}=2^{\,n_Z-1-f_Z}-2^{-f_Z}.
        \]
        Because \(|z_{\min}|\notin
        Z\), the reversible absolute-value operation uses two ancilla qubits,
        \[
            A=A_{\mathrm{sign}}\otimes
            A_{\mathrm{ind}},
        \]
        which store the input sign and indicate the exceptional value \(z_{\min}\),
        respectively.

        The unitary
        \[
            U_a:\mathcal H_Z\otimes\mathcal H_A \longrightarrow \mathcal H_Z\otimes\mathcal H_A
        \]
        acts on clean ancillas as
        \[
            U_a\!\left(
            \ket{z}_Z\ket{0}_{A_{\mathrm{sign}}}\ket{0}_{A_{\mathrm{ind}}} \right) = \ket{\widetilde{|z|}}_Z
            \ket{\sigma(z)}_{A_{\mathrm{sign}}} \ket{\eta(z)}_{A_{\mathrm{ind}}},
        \]
        where
        \[
            \sigma(z)=
            \begin{cases} 0, & z\geq 0, \\ 1, & z<0, \end{cases} \qquad \eta(z)= \begin{cases} 1, & z=z_{\min},
                \\ 0, & z\neq z_{\min},\end{cases}
            \]
            and
            \[
                \widetilde{|z|} = \begin{cases} |z|, & z\neq z_{\min},
                \\ z_{\max}, & z=z_{\min}.\end{cases}
            \]
            This action is extended bijectively to the remaining
            computational basis states.

            A reversible circuit first copies the sign into \(A_{\mathrm{sign}}\), detects \(z_{\min}\) in
            \(A_{\mathrm{ind}}\), and then, for non-exceptional negative inputs, conditionally applies bitwise
            inversion followed by an increment, for example with a QFT-based adder. For \(z=z_{\min}\), the
            output \(z_{\max}\) together with \(A_{\mathrm{ind}}=1\) preserves reversibility. Since \(z_{\min}\)
            is thereby mapped to a large positive value, it remains outside the admissible tolerance region
            \[
                |z|<\varepsilon_{\mathrm{tol}}
            \]
            and cannot be misclassified as a valid solution.
            \subsection*{Reversible Subtraction of the Tolerance Threshold}
            \label{implementation_threshold_operator}

            Let \(N_Z=2^{n_Z}\), and let
            \[
                k'=\sum_{s=0}^{n_Z-1}k'_s2^s, \qquad k'_s\in\{0,1\},
            \]
            be the
            integer encoding of \(-\varepsilon_{\mathrm{tol}}\). Threshold subtraction is the modular constant
            addition
            \[
                U_s:\mathcal H_Z\longrightarrow\mathcal H_Z, \qquad U_s\ket{k}_Z=\ket{k\oplus k'}_Z,
                \qquad k\oplus k':=(k+k')\bmod N_Z.
            \]
            On the admissible computational subspace used by the oracle,
            theregister range is chosen so that this modular operation coincides with ordinary subtraction and
            no wrap-around occurs.

            Using Draper addition,
            \[
                U_s = \mathrm{QFT}_{n_Z}^{-1}U_\phi\mathrm{QFT}_{n_Z}.
            \]
            For the \(r\)-th
            value qubit, define
            \[
                \Pi^{(0)}_{Z_r} = \ket{0}\bra{0}_{Z_r}\otimes I_{\mathrm{rest}}, \qquad
                \Pi^{(1)}_{Z_r} = \ket{1}\bra{1}_{Z_r}\otimes I_{\mathrm{rest}}.
            \]
            The bit \(k'_s\) contributes the
            phase \(2\pi k'_s2^{\,s+r-n_Z}\) to qubit \(Z_r\), so that
            \[
                U_\phi = \prod_{r=0}^{n_Z-1} \left(
                \Pi^{(0)}_{Z_r} + \exp\!\left[ 2\pi i\sum_{s=0}^{n_Z-1}k'_s2^{\,s+r-n_Z} \right] \Pi^{(1)}_{Z_r}
                \right).
            \]
            Thus, \(U_\phi\) has the same phase-accumulation structure as register-based QFT
            addition, except that the controls \(k'_s\) are classical bits of a fixed constant. More general
            out-of-place and register-based constructions are described in
            \cite{VedralBarencoEkert:1996,Draper:2000,RuizPerez:2017}.
            \subsection*{Reversible Residual Operator}
            \label{implementation_residual_operator}
            The residual operator acts on \(\mathcal H_X\otimes\mathcal H_W\otimes\mathcal H_Z\). The spatial
            register encodes grid points \(x_i\in X\) with spacing \(h=2^{-f_X}\) and has sufficient overflow
            capacity to represent\(x_i\pm h\).

            For
            \[
                r(x,w)=D_h^2u(x,w)+f\bigl(x,u(x,w)\bigr),
            \]
            let \(\rho(i,j)\) be the unique index satisfying
            \begin{equation}
                \label{definition_discretized_residual} z_{\rho(i,j)} = r(x_i,w_j) =
                D_h^2u(x_i,w_j)+f\bigl(x_i,u(x_i,w_j)\bigr).
            \end{equation}
            The unitary \(U_r=U_fU_{D^2}\)
            accumulates the encoded residual modulo \(N_Z\):
            \[
                U_r\ket{i}_X\ket{j}_W\ket{k}_Z =
                \ket{i}_X\ket{j}_W\ket{k\oplus\rho(i,j)}_Z,
            \]
            where \(\oplus\) denotes addition modulo \(N_Z\).
            Hence,
            \[
                U_r\ket{i}_X\ket{j}_W\ket{0}_Z = \ket{i}_X\ket{j}_W\ket{\rho(i,j)}_Z.
            \]
            The centered finite-difference approximation is
            \[
                D_h^2u(x_i,w_j) =
                \frac{u(x_i-h,w_j)-2u(x_i,w_j)+u(x_i+h,w_j)}{h^2}.
            \]
            Reversible translations of the spatial
            register provide coherentaccess to \(x_i-h\), \(x_i\), and \(x_i+h\), after which inverse
            translations restore the original spatial encoding. The three contributions are accumulated
            coherently in the value register.

            The ansatz is linear in \(w=(w^0,\ldots,w^{m-1})\) and polynomial in \(x\),
            \[
                u(x,w) =
                \sum_{\ell=0}^{m-1}w^\ell\phi_\ell(x), \qquad \phi_\ell(x) = \sum_{d=0}^{D_\ell}a_{\ell,d}x^d.
            \]
            Consequently, the finite-difference contribution consists of terms of the form \(c\,w^\ell x_i^d\),
            with fixed coefficients \(c\).

            The forcing term is assumed to be polynomial in its second argument,
            \begin{equation}
                \label{def:forcing_term} f(x,u)=\sum_{q=0}^{Q}b_q(x)u^q,
            \end{equation}
            where each \(b_q(x)\) is
            polynomial in \(x\). Substitution of the ansatz therefore yields a polynomial in \(x\) and \(w\),
            including monomials \(w^{j_1}\cdots w^{j_q}\).

            Let \(\delta(i,j)\) and \(\eta(i,j)\) encode \(D_h^2u(x_i,w_j)\) and \(f(x_i,u(x_i,w_j))\),
            respectively. Then
            \[
                U_{D^2}\ket{i}_X\ket{j}_W\ket{k}_Z =
                \ket{i}_X\ket{j}_W\ket{k\oplus\delta(i,j)}_Z,
            \]
            and
            \[
                U_f\ket{i}_X\ket{j}_W\ket{k}_Z =
                \ket{i}_X\ket{j}_W\ket{k\oplus\eta(i,j)}_Z.
            \]
            Thus,
            \[
                \rho(i,j)=\delta(i,j)\oplus\eta(i,j),
            \]
            and
            the composition \(U_r=U_fU_{D^2}\) has the stated action. Since both subroutines are reversible, the
            induced computational-basismap is bijective and therefore unitary.

            \bibliographystyle{plain} 
            \bibliography{bibliothek} 

@article{Grover:1996,
  author    = {Lov K. Grover},
  title     = {A Fast Quantum Mechanical Algorithm for Database Search},
  booktitle = {Proceedings of the 28th Annual ACM Symposium on Theory of Computing (STOC)},
  year      = {1996},
  pages     = {212--219},
  publisher = {ACM},
  doi       = {10.1145/237814.237866},
  url       = {https://doi.org/10.1145/237814.237866}
}

@article{Brassard:2000,
  author    = {Gilles Brassard and Peter H{\o}yer and Michele Mosca and Alain Tapp},
  title     = {Quantum Amplitude Amplification and Estimation},
  journal   = {Contemporary Mathematics},
  volume    = {305},
  pages     = {53--74},
  year      = {2002},
  publisher = {American Mathematical Society},
  note      = {arXiv:quant-ph/0005055}
}

@book{Nielsen:2010,
  title        = {Quantum Computation and Quantum Information},
  author       = {Nielsen, Michael A. and Chuang, Isaac L.},
  year         = {2010},
  edition      = {10th Anniversary Edition},
  publisher    = {Cambridge University Press},
  address      = {Cambridge}
}

@article{Draper:2000,
  author  = {Thomas G. Draper},
  journal = {Contemporary Mathematics},
  title   = {Addition on a Quantum Computer},
  howpublished = {\href{https://doi.org/10.48550/arXiv.quant-ph/0008033}{arXiv:quant-ph/0008033}},
}

@article{RuizPerez:2017,
  author = {Lidia Ruiz-Perez and Juan Carlos Garcia-Escartin},
  title = {Quantum arithmetic with the Quantum Fourier Transform},
  journal = {Quantum Information Processing},
  volume = {16},
  number = {1},
  pages = {152},
  year = {2017},
  doi = {10.1007/s11128-017-1603-1},
  url = {https://arxiv.org/abs/1411.5949}
}

@article{Seidel:2021,
  author = {Raphael Seidel and Nikolay Tcholtchev and Sebastian Bock and Colin Kai-Uwe Becker and Manfred Hauswirth},
  title = {Efficient Floating Point Arithmetic for Quantum Computers},
  journal = {IEEE Access},
  volume = {10},
  pages = {72400--72415},
  year = {2022},
  doi = {10.1109/ACCESS.2022.3187653},
  url = {https://arxiv.org/abs/2112.10537}
}

@article{Wang:2025,
  author = {Siyi Wang and Xiufan Li and Wei Jie Bryan Lee and Suman Deb and Eugene Lim and Anupam Chattopadhyay},
  title = {A Comprehensive Study of Quantum Arithmetic Circuits},
  journal = {Philosophical Transactions A: Mathematical, Physical and Engineering Sciences},
  volume = {383},
  number = {2288},
  pages = {20230392},
  year = {2025},
  doi = {10.1098/rsta.2023.0392},
  url = {https://arxiv.org/abs/2406.03867}
}

@book{Leveque:2007,
  title={Finite Difference Methods for Ordinary and Partial Differential Equations: Steady-State and Time-Dependent Problems},
  author={LeVeque, Randall J.},
  year={2007},
  publisher={SIAM},
  address={Philadelphia, PA}
}

@book{StoerBulirsch:2002,
  author    = {Stoer, J. and Bulirsch, R.},
  title     = {Introduction to Numerical Analysis},
  edition   = {3rd},
  publisher = {Springer},
  year      = {2002}
}

@book{AscherPetzold:1998,
  author    = {Ascher, U. M. and Petzold, L. R.},
  title     = {Computer Methods for Ordinary Differential Equations and Differential-Algebraic Equations},
  publisher = {SIAM},
  year      = {1998}
}

@article{HHL2009,
  author    = {Harrow, Aram W. and Hassidim, Avinatan and Lloyd, Seth},
  title     = {Quantum Algorithm for Linear Systems of Equations},
  journal   = {Physical Review Letters},
  volume    = {103},
  number    = {15},
  pages     = {150502},
  year      = {2009}
}

@article{ChildsKothariSomma2017,
  author    = {Childs, Andrew M. and Kothari, Robin and Somma, Rolando D.},
  title     = {Quantum algorithm for systems of linear equations with exponentially improved dependence on precision},
  journal   = {SIAM Journal on Computing},
  volume    = {46},
  number    = {6},
  pages     = {1920--1950},
  year      = {2017}
}

@article{Raissi:2019,
  author    = {Raissi, M. and Perdikaris, P. and Karniadakis, G. E.},
  title     = {Physics-Informed Neural Networks: A Deep Learning Framework for Solving Forward and Inverse Problems Involving Nonlinear Partial Differential Equations},
  journal   = {Journal of Computational Physics},
  volume    = {378},
  pages     = {686--707},
  year      = {2019}
}

@book{BrennerScott:2008,
  author    = {Susanne C. Brenner and L. Ridgway Scott},
  title     = {The Mathematical Theory of Finite Element Methods},
  edition   = {3},
  series    = {Texts in Applied Mathematics},
  volume    = {15},
  publisher = {Springer},
  year      = {2008},
  doi       = {10.1007/978-0-387-75934-0}
}

@book{Hesthaven:2007,
  author    = {Jan S. Hesthaven and Sigal Gottlieb and David Gottlieb},
  title     = {Spectral Methods for Time-Dependent Problems},
  series    = {Cambridge Monographs on Applied and Computational Mathematics},
  volume    = {21},
  publisher = {Cambridge University Press},
  year      = {2007},
  doi       = {10.1017/CBO9780511618352}
}

@article{DeRyck:2022,
  author  = {Tim De Ryck and Ameya D. Jagtap and Siddhartha Mishra},
  title   = {Error Estimates for Physics-Informed Neural Networks Approximating the Navier--Stokes Equations},
  journal = {IMA Journal of Numerical Analysis},
  volume  = {42},
  number  = {4},
  pages   = {3217--3248},
  year    = {2022},
  doi     = {10.1093/imanum/drab093}
}

@article{DeRyckMishra2024Survey,
  author  = {De Ryck, Tim and Mishra, Siddhartha},
  title   = {Numerical Analysis of Physics-Informed Neural Networks and Related Models in Physics-Informed Machine Learning},
  journal = {Acta Numerica},
  volume  = {33},
  pages   = {73--238},
  year    = {2024},
  doi     = {10.1017/S0962492924000015}
}

@article{ChildsLiu2020,
  author  = {Childs, Andrew M. and Liu, Jin-Peng},
  title   = {Quantum Spectral Methods for Differential Equations},
  journal = {Communications in Mathematical Physics},
  volume  = {375},
  number  = {2},
  pages   = {1427--1457},
  year    = {2020},
  doi     = {10.1007/s00220-019-03528-y}
}

@article{JinLiuYu2022,
  author  = {Jin, Shi and Liu, Nana and Yu, Yue},
  title   = {Quantum Simulation of Partial Differential Equations via Schr{\"o}dingerisation},
  journal = {Journal of Computational Physics},
  volume  = {464},
  pages   = {111282},
  year    = {2022},
  doi     = {10.1016/j.jcp.2022.111282}
}

@article{JinLiuYu2023,
  author  = {Jin, Shi and Liu, Nana and Yu, Yue},
  title   = {Quantum Simulation of Partial Differential Equations: Applications and Extensions},
  journal = {Acta Numerica},
  volume  = {32},
  pages   = {383--478},
  year    = {2023},
  doi     = {10.1017/S0962492923000080}
}

@article{SatoKadowaki2024,
  author  = {Sato, Yuki and Kadowaki, Tadashi},
  title   = {Quantum Algorithms for Partial Differential Equations Based on Hamiltonian Simulation},
  journal = {Physical Review Research},
  volume  = {6},
  number  = {1},
  pages   = {013112},
  year    = {2024},
  doi     = {10.1103/PhysRevResearch.6.013112}
}

@article{LiuEtAl2021,
  author  = {Liu, Jin-Peng and Kolden, H. and Krovi, Hari and Loureiro, N. F. and Trivisa, K. and Childs, Andrew M.},
  title   = {Efficient Quantum Algorithm for Dissipative Nonlinear Differential Equations},
  journal = {Proceedings of the National Academy of Sciences},
  volume  = {118},
  number  = {35},
  pages   = {e2026805118},
  year    = {2021},
  doi     = {10.1073/pnas.2026805118}
}

@article{Krovi2023,
  author  = {Krovi, Hari},
  title   = {Improved Quantum Algorithms for Linear and Nonlinear Differential Equations},
  journal = {Quantum},
  volume  = {7},
  pages   = {913},
  year    = {2023},
  doi     = {10.22331/q-2023-02-02-913}
}

@article{JinLiuYu2023b,
  author  = {Jin, Shi and Liu, Nana and Yu, Yue},
  title   = {Hamiltonian Simulation for Nonlinear Partial Differential Equations},
  journal = {Communications in Computational Physics},
  volume  = {34},
  number  = {1},
  pages   = {148--182},
  year    = {2023},
  doi     = {10.4208/cicp.OA-2022-0217}
}

@article{Panichi2026,
  author  = {Panichi, Lorenzo and Wierichs, David and Meyer, Jakob J. and others},
  title   = {Quantum Physics-Informed Neural Networks},
  journal = {Quantum Machine Intelligence},
  volume  = {8},
  number  = {1},
  pages   = {15},
  year    = {2026},
  doi     = {10.1007/s42484-025-00145-3}
}

@article{Berger2025,
  author  = {Berger, Lukas and Friesdorf, Martin},
  title   = {Variational Quantum PDE Solvers Based on Residual Minimization},
  journal = {Quantum Information Processing},
  volume  = {24},
  number  = {3},
  pages   = {91},
  year    = {2025},
  doi     = {10.1007/s11128-025-04512-z}
}

@article{MontanaroPallister2016, author = {Montanaro, Ashley and Pallister, Sam}, title = {Quantum Algorithms and the Finite Element Method}, journal = {Physical Review A}, volume = {93}, number = {3}, pages = {032324}, year = {2016}, doi = {10.1103/PhysRevA.93.032324} }

@article{ChildsLiuOstrander2021, author = {Childs, Andrew M. and Liu, Jin-Peng and Ostrander, Alex}, title = {High-Precision Quantum Algorithms for Partial Differential Equations}, journal = {Quantum}, volume = {5}, pages = {574}, year = {2021}, doi = {10.22331/q-2021-11-10-574} }

@article{Barenco:1995,
author = {Barenco, Adriano and Bennett, Charles H. and Cleve, Richard and DiVincenzo, David P. and Margolus, Norman and Shor, Peter and Sleator, Tycho and Smolin, John and Weinfurter, Harald},
title = {Elementary Gates for Quantum Computation},
journal = {Physical Review A},
volume = {52},
number = {5},
pages = {3457--3467},
year = {1995},
doi = {10.1103/PhysRevA.52.3457},
eprint = {quant-ph/9503016}
}

@article{VedralBarencoEkert:1996,
  author  = {Vedral, Vlatko and Barenco, Adriano and Ekert, Artur},
  title   = {Quantum Networks for Elementary Arithmetic Operations},
  journal = {Physical Review A},
  volume  = {54},
  number  = {1},
  pages   = {147--153},
  year    = {1996},
  doi     = {10.1103/PhysRevA.54.147}
}
            \end{document}